\documentclass[amsmath,notitlepage,amssymb,a4paper]{revtex4-1}
\usepackage{amsthm}
\usepackage{mathrsfs}
\usepackage{eucal}
\usepackage{slashed}

\usepackage[linktocpage=true, pdfauthor={T. Backdahl},
            pdftitle={Perturbations of Kerr with approximate Killing spinors - Wave version}]{hyperref}

\hypersetup{
    colorlinks=true,
    linkcolor=blue,
    filecolor=blue,      
    urlcolor=blue,
    citecolor=blue
}

\usepackage{orcidlink}

\usepackage[safe]{tipa}

\def\sDiv{\mathscr{D}}
\def\sCurl{\mathscr{C}}
\def\sCurlDagger{\mathscr{C}^\dagger}
\def\sTwist{\mathscr{T}}
\DeclareMathOperator{\tho}{\text{\rm\textthorn}}
\DeclareMathOperator{\edt}{\eth}
\DeclareMathOperator{\thop}{\text{\rm\textthorn}^\prime\negthinspace}
\DeclareMathOperator{\edtp}{\eth^\prime\negthinspace}

\def\sdiv{\text{\texthtd}}
\def\scurl{\text{\texthtc}}
\def\stwist{\text{\texthtt}}

\def\InitSlice{\Sigma_0}

\theoremstyle{plain}
\newtheorem{thm}{Theorem}[section]
\newtheorem{cor}[thm]{Corollary}
\newtheorem{lemma}[thm]{Lemma}

\newtheorem{definition}[thm]{Definition}

\newtheorem{remark}[thm]{Remark}
\newtheorem{assumption}[thm]{Assumption}

\newcounter{step}

\def\goodset{\mathbb{V}_g}
\def\badset{\mathbb{V}_b}

\allowdisplaybreaks[2]

\numberwithin{equation}{section}

\declareslashed{}{-}{0}{0}{\Upsilon}
\declareslashed{}{-}{0}{0}{\phi}
\declareslashed{}{-}{0}{0}{\chi}

\begin{document}

\title{Perturbations of Kerr with approximate Killing spinors -- Wave version}

\author{Thomas B\"ackdahl \orcidlink{0000-0003-3240-2445}} 
\email{thomas.backdahl@chalmers.se}
\affiliation{Mathematical Sciences, Chalmers University of Technology and University of Gothenburg, SE-412~96 Gothenburg, Sweden}

\date{\today}

\begin{abstract}
In this paper we develop a new framework for non-linear perturbations of the Kerr spacetime. This is based on a characterization of the Kerr spacetime in terms of a Killing spinor. On the perturbed spacetime, one can construct an approximation of the Killing spinor. Based on this, a number of quantities are constructed measuring the deviation from Kerr. Evolution equations for these quantities are derived. Approximations of the Killing vectors, the mass and angular momentum parameters etc., are constructed along with a full set of equations for their derivatives. In this setting, we don't need a reference background solution. Instead, we covariantly construct the relevant structures on the perturbed spacetime itself. This can eliminate many issues and allow for a cleaner analysis.
\end{abstract}

\maketitle

\section{Introduction}
In this paper we study non-linear perturbations of the Kerr spacetime. 
This is relevant for the Kerr stability problem \cite{Andersson:2019dwi, Hafner:2019kov, Klainerman:2021qzy, GioriKlainermanSzeftel:WaveEstimatesForKerr, Hafner:2025nrv}, the self-force problem \cite{Pound2020}, numerical relativity and related problems.
A natural way to do this is to compare the perturbed spacetime with a background Kerr spacetime. In \cite{Andersson:2021eqc}, we followed that approach.
This approach is fairly straightforward, but it comes with a number of complications. Foremost, one needs to compare with the correct background, i.e. one would need to know the final mass and angular momentum. These parameters can usually not be determined beforehand, so one would need to determine them iteratively or through careful tracking of the evolution. Even with a good gauge choice there is often some residual gauge freedom that needs to be used to get the correct centre of mass, axis of rotation etc. 

In this paper, we present a different approach. Instead of comparing the perturbed spacetime with a background, we use a characterization of the Kerr spacetime in terms of a Killing spinor. Following the ideas in \cite{Backdahl:2010cy, Backdahl:2010fa, Backdahl:2010eq, Backdahl:2011np} we can construct an approximate Killing spinor with the correct asymptotics on the spacetime. In those papers we were able to prove that if the approximate Killing spinor is an actual Killing spinor, the spacetime must be locally diffeomorphic to Kerr. This can be done independently of coordinate or frame choices.

The idea of this paper is that given an approximate Killing spinor, we can covariantly construct a few spinors that are zero if and only if the spacetime is locally diffeomorphic to Kerr. We will call them \emph{small variables}. With an appropriate evolution equation for the approximate Killing spinor, one would then expect the small variables to decay to zero indicating that the spacetime asymptotes to Kerr. This can be done in a covariant way without prior knowledge of the final Kerr parameters. This can eliminate many of the problems described above.

Here, we are not constructing a Kerr background with its corresponding principal null directions, Killing vectors, TME operators, symmetry operators etc. expressed in coordinates. Instead we use the fact that these structures can be constructed from the Killing spinor in the Kerr spacetime. Therefore, we can construct approximations of these structures in a perturbed spacetime using an approximate Killing spinor instead. These structures can be expressed in terms of spinors that are non-zero on Kerr. We call them \emph{non-small variables}. 

In this paper the evolution of the approximate Killing spinor is chosen to satisfy a wave-like equation with a freely chosen gauge source function. Under suitable conditions, we prove that this gives an evolution system for both small and non-small variables. First we describe this system covariantly in the spacetime using abstract spinor formalism. Then we introduce a family of adapted frames, and describe the system in GHP formalism. Here, we find that all spin coefficients and curvature components can be expressed in terms of the small and non-small variables. Hence, both the frame and all variables satisfies a good evolution system. We also express the system in terms of space spinors. This allows us to separate the actual evolution equations from constraint equations. This splitting is similar to the splitting of the Maxwell equation into a constraint equation and an evolution equation. In the end we get a first order symmetric hyperbolic evolution system.

A system with both small and non-small variables might not be good for estimates, so we extract a subsystem for only the small variables. In this case we treat the non-small variables as given. This is also formulated in GHP, and we see that it also takes a first order symmetric hyperbolic form. 
Importantly, we also derive equations for all derivatives of the non-small variables. This is important because we can define operators like the TME operators, symmetry operators etc. formally identical to the corresponding covariant or GHP expressions in Kerr. The coefficients of these will then only depend on the non-small variables.
When one tries to compute commutators between such operators, equations for all derivatives of the non-small variables are needed. This is crucial in the case with large angular momentum parameter $a$, because we can not use smallness of $a$ to estimate any terms in the commutators.

In this paper we focus on the structure of the equations. Decay estimates will be postponed to subsequent papers. One can also use a different evolution equation for the approximate Killing spinor based on null-geodesics. This can be re-phrased as a special choice of the gauge source function used in this paper, so most equations here will be valid also for that case. However, in that setting most evolution equations will be of transport type. This will also be presented in a forthcoming paper.

All calculations in this paper were done using the \emph{xAct} \cite{xActPackage} suite of packages for \emph{Wolfram Mathematica}. In particular the \emph{SymManipulator} \cite{SymManipulatorPackage}, \emph{SymSpin} \cite{SymSpinPackage}, \emph{SpinFrames} \cite{SpinFramesPackage}, \emph{SpaceSpinors} \cite{SpaceSpinorspackage} and \emph{TexAct} \cite{TexActPackage} packages were further developed and used for this project.

\section{Notation and conventions}
In this paper $(\mathcal{M},g_{ab})$ will denote an orientable
and time orientable, globally hyperbolic vacuum spacetime.
The metric $g_{ab}$ will be taken to have signature $(+,-,-,-)$.
It follows that the spacetime admits a spin structure \cite{geroch1968spinor, geroch1970spinor}. 
We will use 2-spinors following the conventions of \cite{PenRin84}. In particular, $A, B, \ldots$ will 
denote abstract spinorial indices, while $a, b, \ldots$ will denote abstract tensorial indices.
We will also extensively use the GHP formalism \cite{GHP}.

\subsection{Symmetric spinors}
Let $S_{k,l}$ denote the vector bundle of symmetric spinors with $k$ unprimed indices and $l$ primed indices. We will call these spinors \emph{symmetric spinors} of valence $(k,l)$. Furthermore, let $\mathcal{S}_{k,l}$ denote the space of smooth sections of $S_{k,l}$. For simplicity we will assume $C^\infty$ sections, but that assumption can be relaxed.

Any spinor can be decomposed into a set of symmetric spinors \cite[Prop 3.3.54]{PenRin84}. 
Therefore, we will work only with symmetric spinors.
However, the spinorial Levi-Civita covariant derivative $\nabla_{AA'}$ acting on a symmetric spinor will in general not give a symmetric spinor, so we decompose it into four fundamental spinor operators.
\begin{definition}[\protect{\cite[Def 13]{AndBaeBlu14a}}]
For any $\varphi_{A_1\dots A_k}{}^{A_{1}'\dots A_{l}'}\in \mathcal{S}_{k,l}$, we define the operators
$\sDiv_{k,l}:\mathcal{S}_{k,l}\rightarrow \mathcal{S}_{k-1,l-1}$,
$\sCurl_{k,l}:\mathcal{S}_{k,l}\rightarrow \mathcal{S}_{k+1,l-1}$,
$\sCurlDagger_{k,l}:\mathcal{S}_{k,l}\rightarrow \mathcal{S}_{k-1,l+1}$ and
$\sTwist_{k,l}:\mathcal{S}_{k,l}\rightarrow \mathcal{S}_{k+1,l+1}$ as
\begin{subequations}
\begin{align}
(\sDiv_{k,l}\varphi)_{A_1\dots A_{k-1}}{}^{A_1'\dots A_{l-1}'}\equiv{}&
\nabla^{BB'}\varphi_{A_1\dots A_{k-1}B}{}^{A_1'\dots A_{l-1}'}{}_{B'},\\
(\sCurl_{k,l}\varphi)_{A_1\dots A_{k+1}}{}^{A_1'\dots A_{l-1}'}\equiv{}&
\nabla_{(A_1}{}^{B'}\varphi_{A_2\dots A_{k+1})}{}^{A_1'\dots A_{l-1}'}{}_{B'},\\
(\sCurlDagger_{k,l}\varphi)_{A_1\dots A_{k-1}}{}^{A_1'\dots A_{l+1}'}\equiv{}&
\nabla^{B(A_1'}\varphi_{A_1\dots A_{k-1}B}{}^{A_2'\dots A_{l+1}')},\\
(\sTwist_{k,l}\varphi)_{A_1\dots A_{k+1}}{}^{A_1'\dots A_{l+1}'}\equiv{}&
\nabla_{(A_1}{}^{(A_1'}\varphi_{A_2\dots A_{k+1})}{}^{A_2'\dots A_{l+1}')}.
\end{align}
\end{subequations}
\end{definition}
The operators are labelled with the valence numbers of the spinor field it is acting on.

The product of two symmetric spinors is in general not symmetric, so we decompose it into irreducible parts. All parts can be expressed in terms of the following symmetrized product from \cite{Aksteiner:2022fmf}. See also \cite{FRAUENDIENER.1993.0153} for a related construction. 
\begin{definition}[\protect{\cite[Def 1]{Aksteiner:2022fmf}}] \label{def:SymProd}
Let $k,l,n,m,i,j$ be integers with $i \leq min(k,n)$ and $j \leq min(l,m)$. The symmetric product is a bilinear form 
\begin{align}
\underset{k,l}{\overset{i,j}{\odot}}: \mathcal{S}_{k,l} \times \mathcal{S}_{n,m} \to{}& \mathcal{S}_{k+n-2i,l+m-2j}.
\end{align}
For $\phi \in \mathcal{S}_{k,l}, \psi \in \mathcal{S}_{n,m}$, it is given by
\begin{align}
\label{eq:SymMultDef}
(\phi\underset{k,l}{\overset{i,j}{\odot}}\psi)_{A_1 \dots A_{k+n-2i}}^{A'_1 \dots A'_{l+m-2j}}={}& 
\phi_{(A_1 \dots A_{k-i-1}}^{(A'_1 \dots A'_{l-j-1}| B_1 \dots B_i B'_1 \dots B'_j|} 
\psi^{A'_{l-j} \dots A'_{l+m-2j})}_{A_{k-i} \dots A_{k+n-2i})B_1 \dots B_i B'_1 \dots B'_j}
\end{align}
\end{definition} 
Here, the operator is labelled with the number of index contractions $(i,j)$ as well as the valence numbers of the spinor field it is acting on $(k,l)$.
In general we will work with the symmetric product and combinations of the operators above. This will always result in symmetric spinors without any index contractions. Therefore, the spinor indices will not carry any information, so we will usually omit them.

Commutators between the operators $\sDiv, \sCurl, \sCurlDagger, \sTwist$ are given by \cite[Lem~9]{Aksteiner:2022fmf} or \cite[Lem~18]{AndBaeBlu14a}. The Leibniz rule for these operators acting on a symmetric product is given by \cite[Lem~10]{Aksteiner:2022fmf}.
In this notation, the vacuum Bianchi identity is
\begin{align}\label{eq:VacuumBianchiCovariant}
\sCurlDagger_{4,0} \Psi ={}&0.
\end{align}

\subsection{Space spinors}\label{sec:SpaceSpinorsNotation}
To analyse a 3+1 splitting of the equations we will use the \emph{space spinor formalism} developed by Sommers~\cite{Sommers1980a} and Sen~\cite{Sen1981a}. See also section 3 in \cite{Backdahl:2010cy} and chapter 4 in \cite{ValienteKroon2017conformal}.
Let $\tau^{AA'}$ be a future pointing normal to a spacelike hyper-surface with normalization $\tau_{AA'}\tau^{AA'}=2$.
The Sen connection is given by $\nabla_{AB}=\tau_{(A}{}^{A'}\nabla_{B)A'}$. 
The intrinsic connection $D_{AB}$ can be expressed as
\begin{align}
D_{AB}\phi_{C}={}&\nabla_{AB}\phi_{C} - \tfrac{1}{2} K_{ABC}{}^{D} \phi_{D},
\label{eq:IntrinsicToSen}
\end{align}
where $K_{ABCD}$ is the spinor form of the second fundamental form given by
\begin{align}
K_{AB}{}^{CD}={}&\tau_{(A}{}^{A'}\tau^{(C|B'|}\nabla^{D)}{}_{|B'|}\tau_{B)A'}.
\end{align}
We assume that $\tau^{AA'}$ is hyper-surface orthogonal, so the irreducible components of $K_{ABCD}$ are
\begin{align}
\Omega_{ABCD}={}&K_{(ABCD)},&
\slashed{\Omega}={}&K^{AB}{}_{AB}. \label{eq:OmegaDef}
\end{align}
We will also need the normal derivative $\nabla_{\tau} = \tau^{AA'} \nabla_{AA'}$ and acceleration $A=\sCurl_{1,1} \tau $.
The Hermitian conjugate of a spinor in $\mathcal{S}_{k,0}$ is the complex conjugate followed by contraction with $\tau$ on all primed indices. For instance, if $\phi\in \mathcal{S}_{2,0}$, the Hermitian conjugate is
\begin{align}
\widehat{\phi}={}&\tau \underset{1,1}{\overset{0,1}{\odot}}\tau \underset{0,2}{\overset{0,1}{\odot}}\bar{\phi}.
\end{align}

In this paper we will mainly use the Sen connection.
Following Definition~2.9 in \cite{Andersson:2013maa}, we define fundamental space spinor operators.
\begin{definition} \label{def:FundSpaceSpinOp}
Let  $\phi_{A_1\dots A_k}\in \mathcal{S}_{k,0}$. Define the operators
$\sdiv_k: \mathcal{S}_{k,0}\rightarrow \mathcal{S}_{k-2,0}$, 
$\scurl_k: \mathcal{S}_{k,0}\rightarrow \mathcal{S}_{k,0}$ and 
$\stwist_k: \mathcal{S}_{k,0}\rightarrow \mathcal{S}_{k+2,0}$ via
\begin{subequations}
\begin{align}
(\sdiv_k\phi)_{A_1\dots A_{k-2}}\equiv{}&\nabla^{A_{k-1}A_{k}}\phi_{A_1\dots A_k},\\
(\scurl_k\phi)_{A_1\dots A_{k}}\equiv{}&\nabla_{(A_1}{}^{B}\phi_{A_2\dots A_{k})B},\\
(\stwist_k\phi)_{A_1\dots A_{k+2}}\equiv{}&\nabla_{(A_1 A_2}\phi_{A_3\dots A_{k+2})}.
\end{align}
\end{subequations}
These operators are called \emph{Sen-divergence}, \emph{curl} and \emph{twistor operator} respectively.
\end{definition}
Observe that here we define operators in terms of the Sen connection $\nabla_{AB}$ instead of the intrinsic connection $D_{AB}$ used in \cite{Andersson:2013maa}.
Combined with the symmetric product $\odot$, we again get an algebra of symmetric spinors, so we can omit the spinor indices in all covariant equations.

The intrinsic connection is Hermitian. The relation \eqref{eq:IntrinsicToSen} can be used to derive the commutation between the Hermitian conjugate and the fundamental space spinor operators based on the Sen connection. For $\phi\in \mathcal{S}_{k,0}$, we have
\begin{subequations}
\begin{align}
\widehat{\sdiv_{k} \phi}={}&- \sdiv_{k} \widehat{\phi} + (k - 2)\Omega \underset{k,0}{\overset{3,0}{\odot}}\widehat{\phi},\\
\widehat{\scurl_{k} \phi}={}&- \scurl_{k} \widehat{\phi} - \tfrac{1}{6} (k + 2)\slashed{\Omega}\underset{k,0}{\overset{0,0}{\odot}}\widehat{\phi}
 + (k - 1)\Omega \underset{k,0}{\overset{2,0}{\odot}}\widehat{\phi},\\
\widehat{\stwist_{k} \phi}={}&- \stwist_{k} \widehat{\phi} + k\Omega \underset{k,0}{\overset{1,0}{\odot}}\widehat{\phi}.
\end{align}
\end{subequations}
The formal adjoints of $\sdiv_k$, $\scurl_k$, $\stwist_k$  are $-(-1)^k\stwist_{k-2}$, $(-1)^k\scurl_k$ and $-(-1)^k\sdiv_{k+2}$ respectively.
Hermitian adjoints are the Hermitian conjugate of the formal adjoints.

\section{Approximate Killing spinor}
\label{sec:ApproxKillingCovariant}
An important feature of the Kerr spacetime is that it admits a Killing spinor $\kappa\in \mathcal{S}_{2,0}$ satisfying $\sTwist_{2,0} \kappa=0$. This holds for any vacuum Petrov Type D spacetime \cite{walker:penrose:1970CMaPh..18..265W}.
Using $\kappa$, one can construct a vector $\xi =\sCurlDagger_{2,0} \kappa \in \mathcal{S}_{1,1}$. Commuting derivatives, we find that $\xi$ is a Killing vector.
Possibly multiplying $\kappa$ with a complex constant, this vector can be made real in Kerr. 
Together with asymptotic conditions, the Killing spinor can characterize the Kerr spacetime.
\begin{thm}[\protect{\cite[Th 5/B.3]{Backdahl:2010eq}}]
\label{Theorem:CharacterisationKerrData}
A smooth vacuum spacetime $(\mathcal{M},g_{ab})$ is locally
isometric to the Kerr spacetime if and only if the following conditions
are satisfied:
\begin{itemize}
\item[(i)] there exists a Killing spinor $\kappa_{AB}$ such that the
associated Killing vector $\xi_{AA'}$ is real;
\item[(ii)] the spacetime $(\mathcal{M},g_{ab})$ has a stationary
asymptotically flat 4-end with non-vanishing mass in which $\xi_{AA'}$
tends to a time translation.
\end{itemize}
\end{thm}

Assume now that we work on a general asymptotically flat vacuum spacetime and have constructed a good approximation $\kappa$ of the Killing spinor. How this construction should be done will be discussed below. 
\begin{definition}\label{def:HIxi}
Based on an approximate Killing spinor $\kappa\in \mathcal{S}_{2,0}$ we define the following spinors
\begin{subequations}
\begin{align}
H={}&\sTwist_{2,0} \kappa ,\\
\xi ={}&\sCurlDagger_{2,0} \kappa ,\label{eq:xiDef}\\
I\xi ={}&\tfrac{1}{2}\sCurlDagger_{2,0} \kappa
 -  \tfrac{1}{2}\sCurl_{0,2} \bar{\kappa}.
\end{align}
\end{subequations}
\end{definition}
Observe that 
$I\xi$ is $i$ times the imaginary part of $\xi$.
Theorem~\ref{Theorem:CharacterisationKerrData} now says that the spacetime is locally diffeomorphic to Kerr if and only if $H$ and $I\xi$ vanish and $\xi$ has the correct asymptotic behaviour near spacelike infinity. 

We therefore expect that the spacetime is close to Kerr in some sense if $H$ and $I\xi$ are small in some appropriate sense. The details and proof of this expectation will be investigated in subsequent works.
\begin{definition}\label{def:SmallVar}
A quantity that vanishes if $H=0$ and $I\xi=0$, we call a \emph{small variable}. If not, we call it a \emph{non-small variable}.
\end{definition}
In this paper we will derive an evolution system for the small variables $H$, $I\xi$ and some quantities derived from them. In particular we will express the evolution system in the GHP formalism only using the principal null directions of $\kappa$. Some spin coefficients and curvature components can be expressed in terms of the components in this evolution system.
Importantly, we also get a full set of equations for all derivatives of the non-small variables. Namely, the remaining spin coefficients, the remaining curvature component and the non-trivial component of $\kappa$. Furthermore, we derive local expressions that are constant in Kerr containing the mass and angular momentum. From $\kappa$ we also derive a radial variable and approximations of two Killing vectors.
To make sure that all variables evolve properly, we also in the mean time derive a first order symmetric hyperbolic evolution system for the complete set of variables.

On Kerr in GHP notation using the principal null directions of $\kappa$, one can express all relevant structures and operators in terms of the GHP operators and the quantities listed in the previous paragraph. This gives the possibility of proving decay estimates using only the quantities used here and a well chosen foliation of the spacetime. However, the choice of foliation will not be discussed in this paper.

\subsection{Constructing an approximate Killing spinor}
In the series of works \cite{Backdahl:2010cy, Backdahl:2010fa, Backdahl:2010eq, Backdahl:2011np}, the value for $\kappa$ on the initial surface $\InitSlice$ was chosen to satisfy
\begin{align}\label{eq:InitialElliptic}
\mathbf{L}(\kappa)=\sdiv_{4} \stwist_{2} \kappa - 2\Omega \underset{4,0}{\overset{3,0}{\odot}}\stwist_{2} \kappa = 0
\end{align}
along with appropriate asymptotic or boundary conditions. This gives a unique solution on $\InitSlice$. Furthermore, this solution coincides with the actual Killing spinor for the Kerr case.
Furthermore, on $\InitSlice$ the normal derivative was chosen to be 
\begin{align}\label{eq:InitialNormal}
\nabla_{\tau}\kappa = - \scurl_{2} \kappa.
\end{align}
The approximate Killing spinor could then be evolved with the wave-like equation
\begin{align}
0=(\sDiv_{3,1} \sTwist_{2,0} \kappa)_{AB}=\tfrac{2}{3} \nabla^{CC'}\nabla_{CC'}\kappa_{AB} + \tfrac{2}{3} \Psi_{AB}{}^{CD} \kappa_{CD}
\end{align}

In this paper we will, for most parts, use a slightly more general evolution equation
\begin{align}
\sDiv_{3,1} \sTwist_{2,0} \kappa={}&\upsilon, 
\label{eq:kappaPropagation}
\end{align}
where $\upsilon\in \mathcal{S}_{2,0}$ is a gauge source function, i.e. a given function of $\kappa$ and derivatives of $\kappa$. 
There are many possible choices of $\upsilon$, but it is important that $\upsilon=0$ if $H=0$ and $I\xi=0$. 
The simplest choice is $\upsilon=0$. The choice of $\upsilon$ will be discussed further around Assumption~\ref{ass:NiceGaugeSourceFunction}.
\begin{remark}
Similarly, the values of $\mathbf{L}(\kappa)$ and $\nabla_{\tau}\kappa + \scurl_{2} \kappa$ on the initial slice can also be set to some gauge source functions that are zero if $H=0$ and $I\xi=0$. However, for this paper it is an unnecessary complication, so we will refrain from doing that here.
\end{remark}

\subsection{Frames and projection operators}
Classically, on a Petrov Type D spacetime, a tetrad or dyad is chosen to be aligned with the repeated principal null directions of the Weyl tensor or spinor. However, on a general spacetime this is not possible because the Weyl tensor can have up to four principal null directions. The approximate Killing spinor though, has two principal null direction we can use for our frame.
\begin{definition}
Given an approximate Killing spinor $\kappa$, a \emph{principal Killing spinor dyad} is a normalized spin dyad $(o_A, \iota_A)$ such that
\begin{align}
\kappa_{AB}={}&-2 \kappa_{1} o_{(A}\iota_{B)}
\label{eq:kappaDyadRelation},
\end{align}
where $\kappa_1$ a complex scalar of type $\{0,0\}$ and
\begin{align}
o^{A} \bar{o}^{A'} \nabla_{AA'}(\kappa_{1} + \bar{\kappa}_{1'}) < 0.
\label{eq:oOutgoing}
\end{align}
\end{definition}
\bigskip
\begin{remark}\label{rem:dyad}
Observe that:
\begin{itemize}
\item The coefficient $\kappa_{1}$ can be defined covariantly up to sign from 
\begin{align}\label{eq:kappa1Covariant}
\kappa_{1}^2=- \tfrac{1}{2} \kappa_{AB} \kappa^{AB}.
\end{align} 
\item In principle $\kappa_{AB}$ could degenerate at some point, so that \eqref{eq:kappaDyadRelation} is not possible. But in that case $\kappa_{AB} \kappa^{AB}=0$, which outside the black hole region is a large deviation from the value for Kerr $\kappa_{AB} \kappa^{AB}=\frac{-2}{9}(r-i a \cos\theta)^2$. Hence, for small perturbations we can exclude that possibility.
\item The condition \eqref{eq:oOutgoing} encodes that $o^{A} \bar{o}^{A'}$ is an outgoing null direction.
\item All principal Killing spinor dyads are related by spin and boost transformations. Hence, we can work with GHP notation in this class of dyads. In the remainder of this paper all GHP expressions will be with respect to this class of dyads.
\item As we will see below, if the spacetime is Petrov Type D, a principal Killing spinor dyad is also a principal dyad of the Weyl spinor.
\end{itemize}
\end{remark}

Following \cite[Sec IID]{2016arXiv160106084A} we can use the approximate Killing spinor $\kappa$ to covariantly separate components of different spin weight. Observe however that in \cite{2016arXiv160106084A} $\kappa$ was an actual Killing spinor on a vacuum type D spacetime. Due to \eqref{eq:kappaDyadRelation} the algebraic identities in \cite[Sec IID]{2016arXiv160106084A} hold, but not the differential identities.  
\begin{definition}\label{def:SpinProj}
Given $\kappa$, we define spin projection operators for the unprimed indices
\begin{subequations}
\begin{align}
\mathcal{P}^{0}_{4,0}={}&\tfrac{3}{8} \kappa_{1}^{-4}\kappa \underset{2,0}{\overset{0,0}{\odot}}\kappa \underset{0,0}{\overset{0,0}{\odot}}\kappa \underset{2,0}{\overset{2,0}{\odot}}\kappa \underset{4,0}{\overset{2,0}{\odot}},\\
\mathcal{P}^{1}_{4,0}={}&- \kappa_{1}^{-4}\kappa \underset{2,0}{\overset{0,0}{\odot}}\kappa \underset{2,0}{\overset{1,0}{\odot}}\kappa \underset{2,0}{\overset{1,0}{\odot}}\kappa \underset{4,0}{\overset{2,0}{\odot}},\\
\mathcal{P}^{2}_{4,0}={}&\tfrac{1}{16} \kappa_{1}^{-4}\kappa \underset{2,0}{\overset{0,0}{\odot}}\kappa \underset{2,0}{\overset{1,0}{\odot}}\kappa \underset{2,0}{\overset{1,0}{\odot}}\kappa \underset{4,0}{\overset{2,0}{\odot}}
 + \kappa_{1}^{-4}\kappa \underset{4,0}{\overset{1,0}{\odot}}\kappa \underset{4,0}{\overset{1,0}{\odot}}\kappa \underset{4,0}{\overset{1,0}{\odot}}\kappa \underset{4,0}{\overset{1,0}{\odot}} ,\\
\mathcal{P}^{1/2}_{3,1}={}&- \tfrac{3}{4} \kappa_{1}^{-2}\kappa \underset{1,1}{\overset{0,0}{\odot}}\kappa \underset{3,1}{\overset{2,0}{\odot}},\\
\mathcal{P}^{3/2}_{3,1}={}&\kappa_{1}^{-2}\kappa \underset{3,1}{\overset{1,0}{\odot}}\kappa \underset{3,1}{\overset{1,0}{\odot}}
 + \tfrac{1}{12} \kappa_{1}^{-2}\kappa \underset{1,1}{\overset{0,0}{\odot}}\kappa \underset{3,1}{\overset{2,0}{\odot}},\\
\mathcal{P}^{0}_{2,0}={}&- \tfrac{1}{2} \kappa_{1}^{-2}\kappa \underset{0,0}{\overset{0,0}{\odot}}\kappa \underset{2,0}{\overset{2,0}{\odot}},\\
\mathcal{P}^{1}_{2,0}={}&\kappa_{1}^{-2}\kappa \underset{2,0}{\overset{1,0}{\odot}}\kappa \underset{2,0}{\overset{1,0}{\odot}}.
\end{align}
\end{subequations}
\end{definition}
Observe that no component of a valence $(3,1)$ spinor will have spin weight $1/2$ or $3/2$. Here we are only counting the spin weight from the unprimed indices, i.e. letting the primed index remain abstract. We also have the complex conjugate versions for the spin projection of the primed indices.

For illustration, in a principal Killing spinor dyad, $\phi\in \mathcal{S}_{4,0}$ and $\varphi\in \mathcal{S}_{3,1}$, we have
\begin{align*}
(\mathcal{P}^{1}_{4,0} \phi)_{0}={}&0,&
(\mathcal{P}^{1}_{4,0} \phi)_{1}={}&\phi_{1},&
(\mathcal{P}^{1}_{4,0} \phi)_{2}={}&0,&
(\mathcal{P}^{1}_{4,0} \phi)_{3}={}&\phi_{3},&
(\mathcal{P}^{1}_{4,0} \phi)_{4}={}&0,\\
(\mathcal{P}^{1/2}_{3,1} \varphi)_{00'}={}&0,&
(\mathcal{P}^{1/2}_{3,1} \varphi)_{10'}={}&\varphi_{10'},&
(\mathcal{P}^{1/2}_{3,1} \varphi)_{20'}={}&\varphi_{20'},&
(\mathcal{P}^{1/2}_{3,1} \varphi)_{30'}={}&0,\\
(\mathcal{P}^{1/2}_{3,1} \varphi)_{01'}={}&0,&
(\mathcal{P}^{1/2}_{3,1} \varphi)_{11'}={}&\varphi_{11'},&
(\mathcal{P}^{1/2}_{3,1} \varphi)_{21'}={}&\varphi_{21'},&
(\mathcal{P}^{1/2}_{3,1} \varphi)_{31'}={}&0.
\end{align*}

\subsection{Covariant equations}
To analyse the evolution of $H$ and $I\xi$ and the derivatives of the remaining spin coefficients, we define a few extra variables. Here, we define them covariantly, but we will consider their components in the next section.
\begin{definition}
Based on an approximate Killing spinor $\kappa\in \mathcal{S}_{2,0}$, the Weyl spinor $\Psi\in \mathcal{S}_{4,0}$ and Definition~\ref{def:HIxi} we define the following spinors
\begin{subequations}
\begin{align}
\chi ={}&\tfrac{2}{3}\sCurl_{1,1} I\xi ,\\
\psi ={}&\Psi
 -  \mathcal{P}^{0}_{4,0} \Psi ,\label{eq:psiDef}\\
\gamma ={}&\sCurlDagger_{3,1} H,\label{eq:gammaDef}\\
\eta ={}&\kappa_{1}^{-1}\kappa \underset{4,2}{\overset{2,0}{\odot}}\sTwist_{3,1} H,\label{eq:etaDef}\\
\Upsilon ={}&\kappa_{1}^{-2}\sTwist_{0,0} (\kappa_{1}^3\Psi_{2}).
\end{align}
\end{subequations}
\end{definition}
Observe that the curvature component $\Psi_{2}$ can be defined covariantly via
\begin{align}
\Psi_{2}={}&\tfrac{1}{4} \kappa_{1}^{-2}\kappa \underset{2,0}{\overset{2,0}{\odot}}\kappa \underset{4,0}{\overset{2,0}{\odot}}\Psi .
\end{align}
We can also write $\psi$ and $\Upsilon$ in terms of $H$.
\begin{lemma}
The variables $\psi$ and $\Upsilon$ can be expressed in terms of $H$ as
\begin{subequations}
\begin{align}
\psi ={}&\tfrac{3}{4} \kappa_{1}^{-4}\kappa \underset{2,0}{\overset{0,0}{\odot}}\kappa \underset{2,0}{\overset{1,0}{\odot}}\kappa \underset{4,0}{\overset{2,0}{\odot}}\sCurl_{3,1} H
 -  \tfrac{1}{2} \kappa_{1}^{-2}\kappa \underset{4,0}{\overset{1,0}{\odot}}\sCurl_{3,1} H,\label{eq:psiToCurlH}\\
\Upsilon ={}&
 -  \tfrac{1}{2} \kappa_{1}^{-1}\kappa \underset{1,1}{\overset{1,0}{\odot}}\kappa \underset{3,1}{\overset{2,0}{\odot}}\sCurlDagger_{4,0} \psi 
 - \tfrac{9}{4} \Psi_{2} \kappa_{1}^{-1}\kappa \underset{3,1}{\overset{2,0}{\odot}}H.
\label{eq:UpsilonToH}
\end{align}
\end{subequations}
\end{lemma}
\begin{proof}
Commuting derivatives and observing that $\kappa \underset{4,0}{\overset{1,0}{\odot}}\mathcal{P}^{0}_{4,0}=0$, we find
\begin{align}
\sCurl_{3,1} H={}&\sCurl_{3,1} \sTwist_{2,0} \kappa=-2\kappa \underset{4,0}{\overset{1,0}{\odot}}\Psi
=-2\kappa \underset{4,0}{\overset{1,0}{\odot}}\psi .
\label{eq:CurlHEq1}
\end{align}
Applying $\kappa \underset{4,0}{\overset{1,0}{\odot}}$ gives
\begin{align}
\kappa_{1}^{-2}\kappa \underset{4,0}{\overset{1,0}{\odot}}\sCurl_{3,1} H={}&-2\psi
 -  \tfrac{3}{2} \kappa_{1}^{-2}\kappa \underset{2,0}{\overset{0,0}{\odot}}\kappa \underset{4,0}{\overset{2,0}{\odot}}\psi .
\end{align}
Applying $\mathcal{P}^{2}$ and $\mathcal{P}^{1}$ on this as well as $\mathcal{P}^{0}$ on \eqref{eq:psiDef} yield
\begin{align}
\mathcal{P}^{2}_{4,0} \psi ={}&- \tfrac{1}{2} \kappa_{1}^{-2}\mathcal{P}^{2}_{4,0} \kappa \underset{4,0}{\overset{1,0}{\odot}}\sCurl_{3,1} H,&
\mathcal{P}^{1}_{4,0} \psi ={}&-2 \kappa_{1}^{-2}\mathcal{P}^{1}_{4,0} \kappa \underset{4,0}{\overset{1,0}{\odot}}\sCurl_{3,1} H,&
\mathcal{P}^{0}_{4,0} \psi ={}&0.
\end{align}
Summing these parts and using the expressions in Definition~\ref{def:SpinProj}, we get \eqref{eq:psiToCurlH}.
A derivative of \eqref{eq:kappa1Covariant} gives an equation for the gradient of $\kappa_{1}$
\begin{align}
\sTwist_{0,0} \kappa_{1}={}&- \tfrac{1}{2} \kappa_{1}^{-1}\kappa \underset{3,1}{\overset{2,0}{\odot}}H
 -  \tfrac{1}{3} \kappa_{1}^{-1}\kappa \underset{1,1}{\overset{1,0}{\odot}}\xi .
\label{eq:Twistkappa1Eq1}
\end{align}
This together with the definition of $\Upsilon$ yield
\begin{align}
\sTwist_{0,0} \Psi_{2}={}&\kappa_{1}^{-1}\Upsilon
 + \tfrac{3}{2} \Psi_{2} \kappa_{1}^{-2}\kappa \underset{3,1}{\overset{2,0}{\odot}}H
 + \Psi_{2} \kappa_{1}^{-2}\kappa \underset{1,1}{\overset{1,0}{\odot}}\xi .
\label{eq:TwistPsi2Eq1}
\end{align}
The Weyl curvature can be expressed as
\begin{align}
\Psi ={}&\psi + \tfrac{3}{2} \kappa_{1}^{-2}\kappa \underset{2,0}{\overset{0,0}{\odot}}\kappa \underset{0,0}{\overset{0,0}{\odot}}\Psi_{2}.
\label{eq:PsiTopsi}
\end{align}
Applying $\kappa \underset{1,1}{\overset{1,0}{\odot}}\kappa \underset{3,1}{\overset{2,0}{\odot}}\sCurlDagger_{4,0}$ and using the vacuum Bianchi identity \eqref{eq:VacuumBianchiCovariant} and the Leibniz rules, we get
\begin{align}
0={}&\kappa \underset{1,1}{\overset{1,0}{\odot}}\kappa \underset{3,1}{\overset{2,0}{\odot}}\sCurlDagger_{4,0} \psi
 + \tfrac{3}{2}\kappa \underset{1,1}{\overset{1,0}{\odot}}\kappa \underset{3,1}{\overset{2,0}{\odot}}\sCurlDagger_{4,0} (\kappa_{1}^{-2}\kappa \underset{2,0}{\overset{0,0}{\odot}}\kappa \underset{0,0}{\overset{0,0}{\odot}}\Psi_{2})\nonumber\\
={}&\kappa \underset{1,1}{\overset{1,0}{\odot}}\kappa \underset{3,1}{\overset{2,0}{\odot}}\sCurlDagger_{4,0} \psi
 + \tfrac{3}{2} \kappa_{1}^{-2}\kappa \underset{1,1}{\overset{1,0}{\odot}}\kappa \underset{3,1}{\overset{2,0}{\odot}}\sCurlDagger_{4,0} \kappa \underset{2,0}{\overset{0,0}{\odot}}\kappa \underset{0,0}{\overset{0,0}{\odot}}\Psi_{2}
 + 2 \Psi_{2}\kappa \underset{3,1}{\overset{2,0}{\odot}}H
 + \tfrac{4}{3} \Psi_{2}\kappa \underset{1,1}{\overset{1,0}{\odot}}\xi \nonumber\\
={}&\kappa \underset{1,1}{\overset{1,0}{\odot}}\kappa \underset{3,1}{\overset{2,0}{\odot}}\sCurlDagger_{4,0} \psi
 + \tfrac{3}{2} \Psi_{2}\kappa \underset{3,1}{\overset{2,0}{\odot}}H
 - 2 \Psi_{2}\kappa \underset{1,1}{\overset{1,0}{\odot}}\xi
 + 2 \kappa_{1}^2\sTwist_{0,0} \Psi_{2}.
\end{align}
Together with \eqref{eq:TwistPsi2Eq1}, we get \eqref{eq:UpsilonToH}.
\end{proof}

With the notion of Definition~\ref{def:SmallVar} the variables $H$, $I\xi$, $\chi$, $\psi$, $\gamma$, $\eta$ and $\Upsilon$ are \emph{small variables}, while $\kappa_{1}$, $\xi$ and $\Psi_{2}$ are \emph{non-small variables}.
We will now derive a set of evolution equations for a subset of the small variables, in terms of a closed evolution system with coefficients depending on the non-small variables $\kappa_{1}$, $\xi$ and $\Psi_{2}$. 

\begin{lemma}\label{lem:SmallVarPropCovariant}
Under the evolution \eqref{eq:kappaPropagation}, the small variables $H$, $I\xi$, $\chi$ and $\psi$ satisfies the equations
\begin{subequations}
\begin{align}
\sDiv_{3,1} H={}&\upsilon,\label{eq:DivH}\\
\sCurl_{3,1} H={}&-2\kappa \underset{4,0}{\overset{1,0}{\odot}}\psi \label{eq:CurlH},\\
\sDiv_{1,1} I\xi ={}&0,\label{eq:DivIxi}\\
\sCurl_{1,1} I\xi ={}&\tfrac{3}{2}\chi,\label{eq:CurlIxi}\\
\sCurlDagger_{2,0} \chi ={}&- \tfrac{1}{4}\sCurlDagger_{2,0} \upsilon
 -  \tfrac{1}{2}\psi \underset{3,1}{\overset{3,0}{\odot}}H
 -  \tfrac{3}{4} \Psi_{2} \kappa_{1}^{-2}\kappa \underset{1,1}{\overset{1,0}{\odot}}\kappa \underset{3,1}{\overset{2,0}{\odot}}H \nonumber\\
& + \tfrac{1}{4}\sCurl_{0,2} \bar{\upsilon}
 + \tfrac{1}{2}\bar{\psi}\underset{1,3}{\overset{0,3}{\odot}}\bar{H}
 + \tfrac{3}{4} \bar\Psi_{2} \bar{\kappa}_{1'}^{-2}\bar{\kappa}\underset{1,1}{\overset{0,1}{\odot}}\bar{\kappa}\underset{1,3}{\overset{0,2}{\odot}}\bar{H},\label{eq:CurlDgchi}\\
0={}&\mathcal{P}^{3/2}_{3,1} \sCurlDagger_{4,0} \psi
 -  \tfrac{3}{2} \Psi_{2} \kappa_{1}^{-2}\mathcal{P}^{3/2}_{3,1} \kappa \underset{3,1}{\overset{1,0}{\odot}}H.\label{eq:CurlDgpsi}
\end{align}
\end{subequations}
\end{lemma}
\begin{proof}
The equation \eqref{eq:DivH} is just the evolution equation \eqref{eq:kappaPropagation}, while \eqref{eq:CurlH} is \eqref{eq:CurlHEq1}.
The relation \eqref{eq:DivIxi} is a direct consequence of a commutator. \eqref{eq:CurlIxi} is just the definition of $\chi$.
By applying $\sTwist_{1,1}$ to the definition of $I\xi$ and commuting derivatives we get
\begin{align}
\sTwist_{1,1} I\xi ={}&\tfrac{3}{4}\sCurlDagger_{3,1} \sTwist_{2,0} \kappa
 -  \tfrac{3}{4}\sCurl_{1,3} \sTwist_{0,2} \bar{\kappa}
 =\tfrac{3}{4}\gamma -  \tfrac{3}{4}\bar{\gamma}.
\end{align}
Similarly applying $\sDiv_{2,2}$ to the definition of $\gamma$ and commuting derivatives we get
\begin{align}
\sDiv_{2,2} \gamma ={}&\tfrac{1}{2}\sCurlDagger_{2,0} \upsilon
 + \Psi \underset{3,1}{\overset{3,0}{\odot}}H
 =\tfrac{1}{2}\sCurlDagger_{2,0} \upsilon
 + \psi \underset{3,1}{\overset{3,0}{\odot}}H
 + \frac{3 \Psi_{2}}{2 \kappa_{1}^2}\kappa \underset{1,1}{\overset{1,0}{\odot}}\kappa \underset{3,1}{\overset{2,0}{\odot}}H.
\label{eq:Divgamma}
\end{align}
Using the definition of $\chi$ and commuting derivatives we get
\begin{align}
\sCurlDagger_{2,0} \chi ={}&- \tfrac{2}{3}\sDiv_{2,2} \sTwist_{1,1} I\xi
 + \tfrac{1}{2}\sTwist_{0,0} \sDiv_{1,1} I\xi 
=- \tfrac{1}{2}\sDiv_{2,2} \gamma
 + \tfrac{1}{2}\sDiv_{2,2} \bar{\gamma}.
\end{align}
This together with \eqref{eq:Divgamma} gives \eqref{eq:CurlDgchi}.
Applying $\mathcal{P}^{3/2}_{3,1} \sCurlDagger_{4,0}$ on \eqref{eq:PsiTopsi} and using the vacuum Bianchi identity \eqref{eq:VacuumBianchiCovariant}, we get
\begin{align}
0={}&\mathcal{P}^{3/2}_{3,1} \sCurlDagger_{4,0} \Psi
=\mathcal{P}^{3/2}_{3,1} \sCurlDagger_{4,0} \psi
 + \tfrac{3}{2}\mathcal{P}^{3/2}_{3,1} \sCurlDagger_{4,0} (\kappa_{1}^{-2}\kappa \underset{2,0}{\overset{0,0}{\odot}}\kappa \underset{0,0}{\overset{0,0}{\odot}}\Psi_{2}).
\end{align}
Equation \eqref{eq:CurlDgpsi} then follows from the Leibniz rules.
\end{proof}

The non-small variables $\kappa_{1}$, $\xi$ and $\Psi_{2}$ appears as coefficients in our evolution system as well as in all operators that can be written covariantly in terms of the Killing spinor on Kerr. See for instance \cite{2016arXiv160106084A} and \cite{Aksteiner:2016mol} for expressions of the TME and TSI operators, as well as symmetry operators. These covariant expressions can be used to define operators also on the perturbed spacetime. However, commutators between such operators, will not exactly satisfy the same identities as in Kerr. To be able to calculate such commutators, one would need a complete set of equations for all derivatives of the non-small variables. The following lemma provides that.
\begin{lemma}\label{lem:DerNonsmallCovariant}
Under the evolution \eqref{eq:kappaPropagation}, the non-small variables $\kappa_{1}$, $\xi$, $\Psi_{2}$ satisfies the equations
\begin{subequations}
\begin{align}
\sTwist_{0,0} \kappa_{1}={}&- \tfrac{1}{2} \kappa_{1}^{-1}\kappa \underset{3,1}{\overset{2,0}{\odot}}H
 -  \tfrac{1}{3} \kappa_{1}^{-1}\kappa \underset{1,1}{\overset{1,0}{\odot}}\xi, \label{eq:Twistkappa1}\\
\sDiv_{1,1} \xi ={}&0,\label{eq:Divxi}\\
\sCurl_{1,1} \xi ={}&- \tfrac{3}{4}\upsilon
 + \tfrac{3}{2}\kappa \underset{4,0}{\overset{2,0}{\odot}}\psi
 - 3 \Psi_{2}\kappa ,\label{eq:Curlxi}\\
\sCurlDagger_{1,1} \xi ={}&-3\bar{\chi}
 -  \tfrac{3}{4}\bar{\upsilon}
 + \tfrac{3}{2}\bar{\kappa}\underset{0,4}{\overset{0,2}{\odot}}\bar{\psi}
 - 3 \bar\Psi_{2}\bar{\kappa},\label{eq:CurlDgxi}\\
\sTwist_{1,1} \xi ={}&\tfrac{3}{2}\gamma ,\label{eq:Twistxi}\\
\sTwist_{0,0} \Psi_{2}={}&\kappa_{1}^{-1}\Upsilon
 + \tfrac{3}{2} \Psi_{2} \kappa_{1}^{-2}\kappa \underset{3,1}{\overset{2,0}{\odot}}H
 + \Psi_{2} \kappa_{1}^{-2}\kappa \underset{1,1}{\overset{1,0}{\odot}}\xi .
 \label{eq:TwistPsi2}
\end{align}
\end{subequations}
\end{lemma}
\begin{proof}
\eqref{eq:Twistkappa1} is just \eqref{eq:Twistkappa1Eq1}. \eqref{eq:Divxi} is a direct consequence of a commutator. 
Another commutator gives 
\begin{align}
\sCurl_{1,1} \xi ={}&
\sCurl_{1,1} \sCurlDagger_{2,0} \kappa 
=- \tfrac{3}{4}\upsilon
 + \tfrac{3}{2}\Psi \underset{2,0}{\overset{2,0}{\odot}}\kappa .
\end{align}
The relation \eqref{eq:PsiTopsi} then gives \eqref{eq:Curlxi}.
The definition of $\chi$ together with the complex conjugate of \eqref{eq:Curlxi} gives
\begin{align}
\bar{\chi}={}&- \tfrac{2}{3}\sCurlDagger_{1,1} I\xi
=- \tfrac{1}{3}\sCurlDagger_{1,1} \xi
 + \tfrac{1}{3}\sCurlDagger_{1,1} \bar{\xi}
=- \tfrac{1}{3}\sCurlDagger_{1,1} \xi
 -  \tfrac{1}{4}\bar{\upsilon}
 + \tfrac{1}{2}\bar{\kappa}\underset{0,4}{\overset{0,2}{\odot}}\bar{\psi}
 -  \bar\Psi_{2}\bar{\kappa}.
\end{align}
This is equivalent to \eqref{eq:CurlDgxi}. Equation \eqref{eq:Twistxi} follows by commuting derivatives, while \eqref{eq:TwistPsi2} is the same as \eqref{eq:TwistPsi2Eq1}.
\end{proof}
For some arguments, we don't need the full set of equations for the non-small variables. As an alternative we extract an evolution system from it. 
Observe that \eqref{eq:UpsilonToH} and \eqref{eq:TwistPsi2} together gives
\begin{align}
\sTwist_{0,0} \Psi_{2}={}&
 -  \tfrac{1}{2} \kappa_{1}^{-2}\kappa \underset{1,1}{\overset{1,0}{\odot}}\kappa \underset{3,1}{\overset{2,0}{\odot}}\sCurlDagger_{4,0} \psi
- \tfrac{3}{4} \Psi_{2} \kappa_{1}^{-2}\kappa \underset{3,1}{\overset{2,0}{\odot}}H
 + \Psi_{2} \kappa_{1}^{-2}\kappa \underset{1,1}{\overset{1,0}{\odot}}\xi .
\label{eq:TwistPsi2Alt}
\end{align}
Note that \eqref{eq:TwistPsi2Alt} and \eqref{eq:CurlDgpsi} together are equivalent to the full vacuum Bianchi system \eqref{eq:VacuumBianchiCovariant}.
\begin{cor}\label{cor:TotalEvolCovariant}
A total evolution system for the variables $\kappa_1$, $\xi$, $\Psi_{2}$, $H$, $I\xi$, $\chi$ and $\psi$, is given by \eqref{eq:Twistkappa1}, \eqref{eq:Divxi}, \eqref{eq:Curlxi}, \eqref{eq:TwistPsi2Alt} and lemma~\ref{lem:SmallVarPropCovariant}.
\end{cor}
We will see in the next section that this indeed is a well posed evolution system.

\section{GHP formulation}
In this section we present GHP formulations in the class of principal Killing spinor dyads of the equations of the previous section.
For convenience, we define the weighted scalars
\begin{subequations}
\begin{align}
K_{\rho}{}={}&\frac{3 H_{10'}}{2 \kappa_{1}},&
K_{\rho '}{}={}&\frac{3 H_{21'}}{2 \kappa_{1}},&
K_{\tau}{}={}&\frac{3 H_{11'}}{2 \kappa_{1}},&
K_{\tau '}{}={}&\frac{3 H_{20'}}{2 \kappa_{1}},\\
I_{\rho}{}={}&\tfrac{2}{3} I\xi_{00'},&
I_{\tau}{}={}&\tfrac{2}{3} I\xi_{01'},&
I_{\tau '}{}={}&\tfrac{2}{3} I\xi_{10'},&
I_{\rho '}{}={}&\tfrac{2}{3} I\xi_{11'}.
\end{align}
\end{subequations}
The GHP prime operation matches the notation. Observe also that $(\kappa_{1})'=-\kappa_{1}$. Under complex conjugation we have
\begin{align}
\overline{I_{\rho}{}}={}&- I_{\rho}{},&
\overline{I_{\tau}{}}={}&- I_{\tau '}{},&
\overline{I_{\tau '}{}}={}&- I_{\tau}{},&
\overline{I_{\rho '}{}}={}&- I_{\rho '}{}.
\end{align}
The spin coefficients can be expressed in terms of $H$ and $\xi$.
\begin{subequations}
\begin{align}
\kappa ={}&\frac{H_{00'}}{2 \kappa_{1}{}},&
\sigma ={}&\frac{H_{01'}}{2 \kappa_{1}{}},&
\sigma '={}&\frac{H_{30'}}{2 \kappa_{1}{}},&
\kappa '={}&\frac{H_{31'}}{2 \kappa_{1}{}},\\
\rho ={}&\tfrac{1}{3} K_{\rho}{}
 + \frac{\xi_{00'}}{3 \kappa_{1}{}},&
\rho '={}&\tfrac{1}{3} K_{\rho '}{}
 -  \frac{\xi_{11'}}{3 \kappa_{1}{}},&
\tau ={}&\tfrac{1}{3} K_{\tau}{}
 + \frac{\xi_{01'}}{3 \kappa_{1}{}},&
\tau '={}&\tfrac{1}{3} K_{\tau '}{}
 -  \frac{\xi_{10'}}{3 \kappa_{1}{}}. \label{eq:xiGHP}
\end{align}
\end{subequations}
Observe that $\kappa$, $\kappa'$,  $\sigma$ and $\sigma'$ are small variables, while $\rho$, $\rho'$, $\tau$ and $\tau'$ are non-small variables.

We can also express the components of $I\xi$
\begin{subequations}
\begin{align}  
I\xi_{00'}={}&- \tfrac{1}{2} K_{\rho}{} \kappa_{1}
 + \tfrac{1}{2} \overline{K_{\rho}{}} \bar{\kappa}_{1'}
 + \tfrac{3}{2} \kappa_{1} \rho
 -  \tfrac{3}{2} \bar{\kappa}_{1'} \bar{\rho},\\
I\xi_{01'}={}&- \tfrac{1}{2} K_{\tau}{} \kappa_{1}
 -  \tfrac{1}{2} \overline{K_{\tau '}{}} \bar{\kappa}_{1'}
 + \tfrac{3}{2} \kappa_{1} \tau
 + \tfrac{3}{2} \bar{\kappa}_{1'} \bar{\tau}',\\
I\xi_{10'}={}&\tfrac{1}{2} K_{\tau '}{} \kappa_{1}
 + \tfrac{1}{2} \overline{K_{\tau}{}} \bar{\kappa}_{1'}
 -  \tfrac{3}{2} \bar{\kappa}_{1'} \bar{\tau}
 -  \tfrac{3}{2} \kappa_{1} \tau ',\\
I\xi_{11'}={}&\tfrac{1}{2} K_{\rho '}{} \kappa_{1}
 -  \tfrac{1}{2} \overline{K_{\rho '}{}} \bar{\kappa}_{1'}
 -  \tfrac{3}{2} \kappa_{1} \rho '
 + \tfrac{3}{2} \bar{\kappa}_{1'} \bar{\rho}'.
\end{align}
\end{subequations}
The components of $\psi$ are just the Weyl curvature components except the non-small variable $\Psi_2$.
\begin{align}
\psi_{0}={}&\Psi_{0},&
\psi_{1}={}&\Psi_{1},&
\psi_{2}={}&0,&
\psi_{3}={}&\Psi_{3},&
\psi_{4}={}&\Psi_{4}.
\end{align}
This also explains Remark~\ref{rem:dyad}. In the Petrov Type D case, $H=0$, so $\psi=0$ which implies $\Psi_{0}=\Psi_{1}=\Psi_{3}=\Psi_{4}=0$, i.e. the dyad is a principal dyad of the Weyl spinor.

\begin{definition}
Define the following set of small variables
\begin{align}
\goodset={}&\{K_{\rho }{}, K_{\tau }{}, K_{\rho '}{}, K_{\tau '}{}, \kappa, \sigma, \kappa ', \sigma ', I_{\rho }{}, I_{\tau }{}, I_{\rho '}{}, I_{\tau '}{},\Psi_{0},\Psi_{1},\Psi_{3},\Psi_{4},\chi _{0},\chi _{1},\chi _{2}\}
\end{align}
and the set of non-small variables
\begin{align}
\badset={}&\{\kappa _{1}{}, \Psi_{2}, \rho, \rho ', \tau, \tau '\}.
\end{align}
\end{definition}
The components of $\gamma$, $\eta$ and $\Upsilon$ are also small variables, but we can express them as first order derivatives of the $\goodset$ variables. See Appendix~\ref{sec:Appendix} for explicit expressions. The components of $\xi$ are also non-small variables, but they can be expressed in terms of the $\badset$ variables and the first 4 variables in $\goodset$ via  the relation \eqref{eq:xiGHP}.

We can now directly translate all the covariant equations in the previous section to GHP form.
In GHP form the equation~\eqref{eq:DivH} reads
\begin{subequations}
\begin{align}
\frac{3 \upsilon_{0}}{2 \kappa_{1}}={}&(\tho {} - 4 \rho -  \bar{\rho})K_{\tau}{}
 + 3 (\thop {} - 2 \rho ' -  \bar{\rho}')\kappa
 -  (\edt {} - 4 \tau -  \bar{\tau}')K_{\rho}{}
 - 3 (\edtp {} -  \bar{\tau} - 2 \tau ')\sigma
 + 5 K_{\rho '}{} \kappa\nonumber\\
& - 5 K_{\tau '}{} \sigma ,\\
\frac{3 \upsilon_{1}}{2 \kappa_{1}}={}&(\tho {} - 3 \rho -  \bar{\rho})K_{\rho '}{}
 + (\thop {} - 3 \rho ' -  \bar{\rho}')K_{\rho}{}
 -  (\edt {} - 3 \tau -  \bar{\tau}')K_{\tau '}{}
 -  (\edtp {} -  \bar{\tau} - 3 \tau ')K_{\tau}{}
 + 2 K_{\rho}{} K_{\rho '}{}\nonumber\\
& - 2 K_{\tau}{} K_{\tau '}{}
 + 6 \kappa \kappa '
 - 6 \sigma \sigma ',\\
\frac{3 \upsilon_{2}}{2 \kappa_{1}}={}&3 (\tho {} - 2 \rho -  \bar{\rho})\kappa '
 + (\thop {} - 4 \rho ' -  \bar{\rho}')K_{\tau '}{}
 - 3 (\edt {} - 2 \tau -  \bar{\tau}')\sigma '
 -  (\edtp {} -  \bar{\tau} - 4 \tau ')K_{\rho '}{}
 + 5 K_{\rho}{} \kappa '\nonumber\\
& - 5 K_{\tau}{} \sigma '.
\end{align}
\end{subequations}

We will now analyse the GHP form of the evolution systems in Lemma~\ref{lem:SmallVarPropCovariant} and Corollary~\ref{cor:TotalEvolCovariant}.
\begin{lemma}\label{lem:SmallVarPropGHP}
The GHP form of Lemma~\ref{lem:SmallVarPropCovariant} is
\begin{subequations}
\begin{align}
(\thop {} -  \rho ' -  \bar{\rho}')K_{\rho}{}={}&(\edtp {} -  \bar{\tau} -  \tau ')K_{\tau}{}
 -  K_{\rho}{} K_{\rho '}{}
 + K_{\tau}{} K_{\tau '}{}
 + \frac{3 \upsilon_{1}}{4 \kappa_{1}}
 - 3 \kappa \kappa '
 + 2 K_{\rho '}{} \rho
 + 3 \sigma \sigma '
 - 2 K_{\tau '}{} \tau , \label{eq:GHPPropagation1a}\\
(\tho {} -  \rho -  \bar{\rho})K_{\tau}{}={}&(\edt {} -  \tau -  \bar{\tau}')K_{\rho}{}
 + \tfrac{3}{2} \Psi_{1}
 + \frac{3 \upsilon_{0}}{8 \kappa_{1}}
 - 2 K_{\rho '}{} \kappa
 + 3 \kappa \rho '
 + 2 K_{\tau '}{} \sigma
 - 3 \sigma \tau ',\\
(\tho {} -  \rho -  \bar{\rho})K_{\rho '}{}={}&(\edt {} -  \tau -  \bar{\tau}')K_{\tau '}{}
 -  K_{\rho}{} K_{\rho '}{}
 + K_{\tau}{} K_{\tau '}{}
 + \frac{3 \upsilon_{1}}{4 \kappa_{1}}
 - 3 \kappa \kappa '
 + 2 K_{\rho}{} \rho '
 + 3 \sigma \sigma '
 - 2 K_{\tau}{} \tau ',\\
(\thop {} -  \rho ' -  \bar{\rho}')K_{\tau '}{}={}&(\edtp {} -  \bar{\tau} -  \tau ')K_{\rho '}{}
 + \tfrac{3}{2} \Psi_{3}
 + \frac{3 \upsilon_{2}}{8 \kappa_{1}}
 - 2 K_{\rho}{} \kappa '
 + 3 \kappa ' \rho
 + 2 K_{\tau}{} \sigma '
 - 3 \sigma ' \tau ,\label{eq:GHPPropagation1d}\\
(\thop {} -  \rho ' -  \bar{\rho}')\kappa ={}&(\edtp {} -  \bar{\tau} -  \tau ')\sigma
 -  \tfrac{1}{2} \Psi_{1}
 + \frac{3 \upsilon_{0}}{8 \kappa_{1}}
 -  K_{\rho '}{} \kappa
 + K_{\tau}{} \rho
 + K_{\tau '}{} \sigma
 -  K_{\rho}{} \tau ,\\
(\tho {} -  \rho -  \bar{\rho})\sigma ={}&(\edt {} -  \tau -  \bar{\tau}')\kappa
 + \Psi_{0},\\
(\tho {} -  \rho -  \bar{\rho})\kappa '={}&(\edt {} -  \tau -  \bar{\tau}')\sigma '
 -  \tfrac{1}{2} \Psi_{3}
 + \frac{3 \upsilon_{2}}{8 \kappa_{1}}
 -  K_{\rho}{} \kappa '
 + K_{\tau '}{} \rho '
 + K_{\tau}{} \sigma '
 -  K_{\rho '}{} \tau ',\\
(\thop {} -  \rho ' -  \bar{\rho}')\sigma '={}&(\edtp {} -  \bar{\tau} -  \tau ')\kappa '
 + \Psi_{4},\\
(\thop {} -  \bar{\rho}')I_{\rho}{}={}&(\edtp {} -  \bar{\tau})I_{\tau}{}
 + I_{\rho '}{} \rho
 -  I_{\tau '}{} \tau
 + \chi_{1},\\
(\tho {} -  \bar{\rho})I_{\tau}{}={}&(\edt {} -  \bar{\tau}')I_{\rho}{}
 -  I_{\rho '}{} \kappa
 + I_{\tau '}{} \sigma
 -  \chi_{0},\\
(\tho {} -  \bar{\rho})I_{\rho '}{}={}&(\edt {} -  \bar{\tau}')I_{\tau '}{}
 + I_{\rho}{} \rho '
 -  I_{\tau}{} \tau '
 -  \chi_{1},\\
(\thop {} -  \bar{\rho}')I_{\tau '}{}={}&(\edtp {} -  \bar{\tau})I_{\rho '}{}
 -  I_{\rho}{} \kappa '
 + I_{\tau}{} \sigma '
 + \chi_{2},\label{eq:GHPPropagation1l}\\
(\tho {} - 4 \rho)\Psi_{1}={}&(\edtp {} -  \tau ')\Psi_{0}
 - 3 \Psi_{2} \kappa ,\label{eq:GHPPropagation1m}\\
(\thop {} -  \rho ')\Psi_{0}={}&(\edt {} - 4 \tau)\Psi_{1}
 + 3 \Psi_{2} \sigma ,\\
(\tho {} -  \rho)\Psi_{4}={}&(\edtp {} - 4 \tau ')\Psi_{3}
 + 3 \Psi_{2} \sigma ',\\
(\thop {} - 4 \rho ')\Psi_{3}={}&(\edt {} -  \tau)\Psi_{4}
 - 3 \Psi_{2} \kappa ', \label{eq:GHPPropagation1p}\\
(\tho {} - 2 \rho)\chi_{1}={}&
  (\edtp {} -  \tau ')\chi_{0}
- \tfrac{1}{4} (\tho {} - 2 \rho)\upsilon_{1}
 + \tfrac{1}{4} (\tho {} - 2 \bar{\rho})\bar{\upsilon}_{1'}
 -  \tfrac{1}{4} (\edt {} -  \bar{\tau}')\bar{\upsilon}_{0'}
 + \tfrac{1}{4} (\edtp {} -  \tau ')\upsilon_{0}
 \nonumber\\
& + K_{\tau '}{} \Psi_{1} \kappa_{1}
 -  K_{\rho}{} \Psi_{2} \kappa_{1}
 -  \overline{K_{\tau '}{}} \bar\Psi_{1} \bar{\kappa}_{1'}
 + \overline{K_{\rho}{}} \bar\Psi_{2} \bar{\kappa}_{1'}
 -  \tfrac{1}{4} \upsilon_{2} \kappa
 + \Psi_{3} \kappa_{1} \kappa
 + \tfrac{1}{4} \bar{\upsilon}_{2'} \bar{\kappa}\nonumber\\
& -  \bar\Psi_{3} \bar{\kappa}_{1'} \bar{\kappa}
 -  \Psi_{0} \kappa_{1} \sigma '
 + \bar\Psi_{0} \bar{\kappa}_{1'} \bar{\sigma}'
 -  \kappa \chi_{2},\label{eq:GHPPropagation1q}\\
(\thop {} -  \rho ')\chi_{0}={}&
  (\edt {} - 2 \tau)\chi_{1}
 - \tfrac{1}{4} (\tho {} -  \bar{\rho})\bar{\upsilon}_{2'}
 -  \tfrac{1}{4} (\thop {} -  \rho ')\upsilon_{0}
 + \tfrac{1}{4} (\edt {} - 2 \tau)\upsilon_{1}
 + \tfrac{1}{4} (\edt {} - 2 \bar{\tau}')\bar{\upsilon}_{1'}\nonumber\\
& -  K_{\rho '}{} \Psi_{1} \kappa_{1}
 + K_{\tau}{} \Psi_{2} \kappa_{1}
 + \overline{K_{\tau '}{}} \bar\Psi_{2} \bar{\kappa}_{1'}
 -  \overline{K_{\rho}{}} \bar\Psi_{3} \bar{\kappa}_{1'}
 + \bar\Psi_{4} \bar{\kappa}_{1'} \bar{\kappa}
 + \Psi_{0} \kappa_{1} \kappa '
 + \tfrac{1}{4} \upsilon_{2} \sigma\nonumber\\
& -  \Psi_{3} \kappa_{1} \sigma
 + \tfrac{1}{4} \bar{\upsilon}_{0'} \bar{\sigma}'
 -  \bar\Psi_{1} \bar{\kappa}_{1'} \bar{\sigma}'
 + \sigma \chi_{2},\\
(\tho {} -  \rho)\chi_{2}={}&
  (\edtp {} - 2 \tau ')\chi_{1}
 - \tfrac{1}{4} (\tho {} -  \rho)\upsilon_{2}
 -  \tfrac{1}{4} (\thop {} -  \bar{\rho}')\bar{\upsilon}_{0'}
 + \tfrac{1}{4} (\edtp {} - 2 \bar{\tau})\bar{\upsilon}_{1'}
 + \tfrac{1}{4} (\edtp {} - 2 \tau ')\upsilon_{1}
 \nonumber\\
& + K_{\tau '}{} \Psi_{2} \kappa_{1}
 -  K_{\rho}{} \Psi_{3} \kappa_{1}
 -  \overline{K_{\rho '}{}} \bar\Psi_{1} \bar{\kappa}_{1'}
 + \overline{K_{\tau}{}} \bar\Psi_{2} \bar{\kappa}_{1'}
 + \Psi_{4} \kappa_{1} \kappa
 + \bar\Psi_{0} \bar{\kappa}_{1'} \bar{\kappa}'
 + \tfrac{1}{4} \bar{\upsilon}_{2'} \bar{\sigma}\nonumber\\
& -  \bar\Psi_{3} \bar{\kappa}_{1'} \bar{\sigma}
 + \tfrac{1}{4} \upsilon_{0} \sigma '
 -  \Psi_{1} \kappa_{1} \sigma '
 + \sigma ' \chi_{0},\\
(\thop {} - 2 \rho ')\chi_{1}={}&
  (\edt {} -  \tau)\chi_{2}
- \tfrac{1}{4} (\thop {} - 2 \rho ')\upsilon_{1}
 + \tfrac{1}{4} (\thop {} - 2 \bar{\rho}')\bar{\upsilon}_{1'}
 + \tfrac{1}{4} (\edt {} -  \tau)\upsilon_{2}
 -  \tfrac{1}{4} (\edtp {} -  \bar{\tau})\bar{\upsilon}_{2'}\nonumber\\
& -  K_{\rho '}{} \Psi_{2} \kappa_{1}
 + K_{\tau}{} \Psi_{3} \kappa_{1}
 + \overline{K_{\rho '}{}} \bar\Psi_{2} \bar{\kappa}_{1'}
 -  \overline{K_{\tau}{}} \bar\Psi_{3} \bar{\kappa}_{1'}
 -  \tfrac{1}{4} \upsilon_{0} \kappa '
 + \Psi_{1} \kappa_{1} \kappa '
 + \tfrac{1}{4} \bar{\upsilon}_{0'} \bar{\kappa}'\nonumber\\
& -  \bar\Psi_{1} \bar{\kappa}_{1'} \bar{\kappa}'
 -  \Psi_{4} \kappa_{1} \sigma
 + \bar\Psi_{4} \bar{\kappa}_{1'} \bar{\sigma}
 -  \kappa ' \chi_{0}. \label{eq:GHPPropagation1t}
\end{align}
\end{subequations}
\end{lemma}
\begin{lemma}\label{lem:NonSmallPropGHP}
The GHP form of the equations \eqref{eq:Twistkappa1}, \eqref{eq:Divxi}, \eqref{eq:Curlxi}, \eqref{eq:TwistPsi2Alt} are
\begin{subequations}
\begin{align}
(\tho {} + \rho)\kappa_{1}{}={}&K_{\rho}{} \kappa_{1}{},\label{eq:GHPPropagationkappa1a}\\
(\edt {} + \tau)\kappa_{1}{}={}&K_{\tau}{} \kappa_{1}{},\\
(\edtp {} + \tau ')\kappa_{1}{}={}&K_{\tau '}{} \kappa_{1}{},\\
(\thop {} + \rho ')\kappa_{1}{}={}&K_{\rho '}{} \kappa_{1}{},\label{eq:GHPPropagationkappa1d}\\
(\thop {} -  \bar{\rho}')\xi_{10'}={}&(\edtp {} -  \bar{\tau})\xi_{11'}
 -  \tfrac{3}{4} \upsilon_{2}
 - 3 \Psi_{3} \kappa_{1}{}
 -  \kappa ' \xi_{00'}
 + \xi_{01'} \sigma ',\label{eq:GHPPropagationxi10}\\
(\tho {} -  \bar{\rho})\xi_{11'}={}&(\edt {} -  \bar{\tau}')\xi_{10'}
 + \tfrac{3}{4} \upsilon_{1}
 + 3 \Psi_{2} \kappa_{1}{}
 + \xi_{00'} \rho '
 -  \xi_{01'} \tau ',\\
(\thop {} -  \bar{\rho}')\xi_{00'}={}&(\edtp {} -  \bar{\tau})\xi_{01'}
 -  \tfrac{3}{4} \upsilon_{1}
 - 3 \Psi_{2} \kappa_{1}{}
 + \xi_{11'} \rho
 -  \xi_{10'} \tau ,\\
(\tho {} -  \bar{\rho})\xi_{01'}={}&(\edt {} -  \bar{\tau}')\xi_{00'}
 + \tfrac{3}{4} \upsilon_{0}
 + 3 \Psi_{1} \kappa_{1}{}
 -  \kappa \xi_{11'}
 + \xi_{10'} \sigma ,\label{eq:GHPPropagationxi01}\\
(\tho {} - 3 \rho)\Psi_{2}={}&(\edtp {} - 2 \tau ')\Psi_{1}
 - 2 \Psi_{3} \kappa
 + \Psi_{0} \sigma ',\label{eq:GHPPropagationThoPsi2}\\
(\thop {} - 2 \rho ')\Psi_{1}={}&(\edt {} - 3 \tau)\Psi_{2}
 -  \Psi_{0} \kappa '
 + 2 \Psi_{3} \sigma ,\\
(\tho {} - 2 \rho)\Psi_{3}={}&(\edtp {} - 3 \tau ')\Psi_{2}
 -  \Psi_{4} \kappa
 + 2 \Psi_{1} \sigma ',\\
(\thop {} - 3 \rho ')\Psi_{2}={}&(\edt {} - 2 \tau)\Psi_{3}
 - 2 \Psi_{1} \kappa '
 + \Psi_{4} \sigma .\label{eq:GHPPropagationThopPsi2}
\end{align}
\end{subequations}
where the spin coefficients $\rho$, $\rho'$, $\tau$ and $\tau'$ are given by \eqref{eq:xiGHP}.
\end{lemma}
First we will see that the total evolution system Lemma~\ref{lem:SmallVarPropGHP} and Lemma~\ref{lem:NonSmallPropGHP} together gives a well posed evolution for a good choice of gauge source function $\upsilon$. Here it is important that the $\upsilon$ terms in the system are not expressed as derivatives of the $\goodset$ variables because that could spoil the symmetric hyperbolic structure. Therefore, we make the following assumption.

\begin{assumption}\label{ass:NiceGaugeSourceFunction}
Assume that the gauge source function $\upsilon=0$ or more generally that the first order derivatives of $\upsilon$ in \eqref{eq:GHPPropagation1q}-\eqref{eq:GHPPropagation1t} and $\upsilon$ itself can be written in terms of the $\goodset$ variables without derivatives.
\end{assumption}
For some gauge source functions a modification of $I\xi$ can help compensating for the first order derivatives of $\upsilon$ in \eqref{eq:GHPPropagation1q}-\eqref{eq:GHPPropagation1t}. For cases where Assumption~\ref{ass:NiceGaugeSourceFunction} is not satisfied, one would have to redo this analysis case by case. 

Observe that the equations \eqref{eq:GHPPropagation1a}-\eqref{eq:GHPPropagation1l} and \eqref{eq:GHPPropagationxi10}-\eqref{eq:GHPPropagationxi01}
come in pairs with principal parts of the form
\begin{align}
\thop \varphi={}&\edtp \phi,&
\tho \phi={}&\edt \varphi \label{eq:BianchiPair}
\intertext{or the complex conjugate version}
\thop \varphi={}&\edt \phi,&
\tho \phi={}&\edtp \varphi. \label{eq:BianchiPairAlt}
\end{align}
Such pairs are first order symmetric hyperbolic. 
In fact the GHP form of the massless Dirac equation $\sCurlDagger_{1,0} \phi =0$, for $\phi\in\mathcal{S}_{1,0}$ has this form
\begin{align}
(\thop {} -  \rho ')\phi_{0}={}&(\edt {} -  \tau)\phi_{1}, &
(\tho {} -  \rho)\phi_{1}={}&(\edtp {} -  \tau ')\phi_{0}.
\end{align}

The equations \eqref{eq:GHPPropagation1m}-\eqref{eq:GHPPropagation1p}, \eqref{eq:GHPPropagationThoPsi2}-\eqref{eq:GHPPropagationThopPsi2} is just the standard vacuum Bianchi equations. Given any spacelike foliation, this splits into 3 constraint equations and 5 propagation equations of symmetric hyperbolic form. Similarly, \eqref{eq:GHPPropagationkappa1a}-\eqref{eq:GHPPropagationkappa1d} splits into 3 constraint equations and 1 propagation for $\kappa_1$.
Under Assumption~\ref{ass:NiceGaugeSourceFunction}, the principal part of \eqref{eq:GHPPropagation1q}-\eqref{eq:GHPPropagation1t} is the same as the Maxwell equation.
Also this splits into 1 constraint equation and 3 propagation equations of symmetric hyperbolic form. See the next section for details.
Hence, with the observation that the spin coefficients $\rho$, $\rho'$, $\tau$ and $\tau'$ are given by \eqref{eq:xiGHP} we have the following corollary.
\begin{cor}\label{cor:TotalEvol}
Under Assumption~\ref{ass:NiceGaugeSourceFunction}, Lemma~\ref{lem:SmallVarPropGHP} and Lemma~\ref{lem:NonSmallPropGHP} gives a closed first order symmetric hyperbolic evolution system for all the variables
$$\{K_{\rho }{}, K_{\tau }{}, K_{\rho '}{}, K_{\tau '}{}, \kappa, \sigma, \kappa ', \sigma ', I_{\rho }{}, I_{\tau }{}, I_{\rho '}{}, I_{\tau '}{},\Psi_{0},\Psi_{1},\Psi_{2},\Psi_{3},\Psi_{4},\chi _{0},\chi _{1},\chi _{2}, \kappa _{1}{}, \xi_{00'}, \xi_{01'}, \xi_{10'}, \xi_{11'}\}.$$
Therefore, it gives a well posed evolution.
\end{cor}

For estimates, the total system is not ideal, because not all variables in the system will decay. However, if we treat the non-small variables $\badset$ as given coefficients, we can study just the evolution system given by Lemma~\ref{lem:SmallVarPropGHP}. 
The argument is the same as above, but instead of the full Bianchi system, we have just have \eqref{eq:GHPPropagation1m}-\eqref{eq:GHPPropagation1p} which are of the form \eqref{eq:BianchiPairAlt}. Hence, we have the following corollary.
\begin{cor}\label{cor:SmallEvol}
Under Assumption~\ref{ass:NiceGaugeSourceFunction}, Lemma~\ref{lem:SmallVarPropGHP} gives a first order symmetric hyperbolic evolution system for the small variables
$\goodset$ with coefficients depending on the non-small variables~$\badset$.
\end{cor}

Importantly, we not only get evolution equations for the $\badset$ variables. We get expressions for all derivatives of them.
\begin{lemma}\label{lem:AllGHPDerNonSmall}
All derivatives of the non-small variables $\{\kappa _{1}{},\Psi_{2},\rho ,\rho ',\tau ,\tau '\}$ are
\begin{subequations}
\begin{align}
\tho \kappa_{1}={}&\kappa_{1} (K_{\rho}{} -  \rho),\label{eq:Thokappa1}\\
\thop \kappa_{1}={}&\kappa_{1} (K_{\rho '}{} -  \rho '),\\
\edt \kappa_{1}={}&\kappa_{1} (K_{\tau}{} -  \tau),\\
\edtp \kappa_{1}={}&\kappa_{1} (K_{\tau '}{} -  \tau '),\label{eq:Ethpkappa1}\\
\tho \Psi_{2}={}&-3 K_{\rho}{} \Psi_{2}
 + 3 \Psi_{2} \rho
 + \kappa_{1}^{-1} \Upsilon_{00'},\label{eq:ThoPsi2}\\
\thop \Psi_{2}={}&-3 K_{\rho '}{} \Psi_{2}
 + 3 \Psi_{2} \rho '
 + \kappa_{1}^{-1} \Upsilon_{11'},\\
\edt \Psi_{2}={}&-3 K_{\tau}{} \Psi_{2}
 + 3 \Psi_{2} \tau
 + \kappa_{1}^{-1} \Upsilon_{01'},\\
\edtp \Psi_{2}={}&-3 K_{\tau '}{} \Psi_{2}
 + 3 \Psi_{2} \tau '
 + \kappa_{1}^{-1} \Upsilon_{10'},\label{eq:EthpPsi2}\\
\tho \rho ={}&\tfrac{1}{8} (3 \gamma_{00'} - 2 \eta_{00'}) \kappa_{1}^{-1}
 -  K_{\tau '}{} \kappa
 -  K_{\rho}{} \rho
 + \rho^2
 -  \bar{\kappa} \tau ,\label{eq:Thorho}\\
\thop \rho ={}&- \tfrac{1}{2} \Psi_{2}
 -  \kappa \kappa '
 -  K_{\rho '}{} \rho
 + \rho \rho '
 -  K_{\tau '}{} \tau
 -  \tau \bar{\tau}
 + \tau \tau '\nonumber\\*
& -  \tfrac{1}{8} \kappa_{1}^{-1} (\bar{\upsilon}_{1'} - 6 \gamma_{11'} + 2 \eta_{11'} + 4 \bar\Psi_{2} \bar{\kappa}_{1'} + 4 \bar{\chi}_{1'}),\\
\edt \rho ={}&- \tfrac{3}{4} \Psi_{1}
 -  \tfrac{1}{16} (3 \upsilon_{0} - 6 \gamma_{01'} + 4 \eta_{01'}) \kappa_{1}^{-1}
 -  K_{\tau}{} \rho
 -  \kappa \rho '
 -  K_{\tau '}{} \sigma
 + \rho \tau
 -  \bar{\rho} \tau
 + \sigma \tau ',\\
\edtp \rho ={}&-2 K_{\tau '}{} \rho
 -  \kappa \sigma '
 -  \bar{\sigma} \tau
 + 2 \rho \tau '
 -  \tfrac{1}{8} \kappa_{1}^{-1} (\bar{\upsilon}_{0'} - 6 \gamma_{10'} + 2 \eta_{10'} + 4 \bar\Psi_{1} \bar{\kappa}_{1'} + 4 \bar{\chi}_{0'}),\\
\tho \rho '={}&- \tfrac{1}{2} \Psi_{2}
 -  \kappa \kappa '
 -  K_{\rho}{} \rho '
 + \rho \rho '
 -  K_{\tau}{} \tau '
 + \tau \tau '
 -  \tau ' \bar{\tau}'\nonumber\\*
& -  \tfrac{1}{8} \kappa_{1}^{-1} (\bar{\upsilon}_{1'} + 6 \gamma_{11'} + 2 \eta_{11'} + 4 \bar\Psi_{2} \bar{\kappa}_{1'} + 4 \bar{\chi}_{1'}),\\
\thop \rho '={}&- \tfrac{1}{8} (3 \gamma_{22'} + 2 \eta_{22'}) \kappa_{1}^{-1}
 -  K_{\tau}{} \kappa '
 -  K_{\rho '}{} \rho '
 + \rho '^2
 -  \bar{\kappa}' \tau ',\\
\edt \rho '={}&-2 K_{\tau}{} \rho '
 -  \kappa ' \sigma
 + 2 \rho ' \tau
 -  \bar{\sigma}' \tau '
 -  \tfrac{1}{8} \kappa_{1}^{-1} (\bar{\upsilon}_{2'} + 6 \gamma_{12'} + 2 \eta_{12'} + 4 \bar\Psi_{3} \bar{\kappa}_{1'} + 4 \bar{\chi}_{2'}),\\
\edtp \rho '={}&- \tfrac{3}{4} \Psi_{3}
 -  \tfrac{1}{16} (3 \upsilon_{2} + 6 \gamma_{21'} + 4 \eta_{21'}) \kappa_{1}^{-1}
 -  \kappa ' \rho
 -  K_{\tau '}{} \rho '
 -  K_{\tau}{} \sigma '
 + \sigma ' \tau
 + \rho ' \tau '
 -  \bar{\rho}' \tau ',\\
\tho \tau ={}&\tfrac{3}{4} \Psi_{1}
 + \tfrac{1}{16} (3 \upsilon_{0} + 6 \gamma_{01'} - 4 \eta_{01'}) \kappa_{1}^{-1}
 -  K_{\rho '}{} \kappa
 + \kappa \rho '
 -  K_{\rho}{} \tau
 + \rho \tau
 -  \sigma \tau '
 -  \rho \bar{\tau}',\\
\thop \tau ={}&- \bar{\kappa}' \rho
 -  \kappa ' \sigma
 - 2 K_{\rho '}{} \tau
 + 2 \rho ' \tau
 -  \tfrac{1}{8} \kappa_{1}^{-1} (\bar{\upsilon}_{2'} - 6 \gamma_{12'} + 2 \eta_{12'} + 4 \bar\Psi_{3} \bar{\kappa}_{1'} + 4 \bar{\chi}_{2'}),\\
\edt \tau ={}&\tfrac{1}{8} (3 \gamma_{02'} - 2 \eta_{02'}) \kappa_{1}^{-1}
 -  K_{\rho '}{} \sigma
 -  \rho \bar{\sigma}'
 -  K_{\tau}{} \tau
 + \tau^2,\\
\edtp \tau ={}&\tfrac{1}{2} \Psi_{2}
 -  K_{\rho '}{} \rho
 + \rho \rho '
 -  \rho \bar{\rho}'
 -  \sigma \sigma '
 -  K_{\tau '}{} \tau
 + \tau \tau '\nonumber\\*
& -  \tfrac{1}{8} \kappa_{1}^{-1} (\bar{\upsilon}_{1'} - 6 \gamma_{11'} + 2 \eta_{11'} + 4 \bar\Psi_{2} \bar{\kappa}_{1'} + 4 \bar{\chi}_{1'}),\\
\tho \tau '={}&- \bar{\kappa} \rho '
 -  \kappa \sigma '
 - 2 K_{\rho}{} \tau '
 + 2 \rho \tau '
 -  \tfrac{1}{8} \kappa_{1}^{-1} (\bar{\upsilon}_{0'} + 6 \gamma_{10'} + 2 \eta_{10'} + 4 \bar\Psi_{1} \bar{\kappa}_{1'} + 4 \bar{\chi}_{0'}),\\
\thop \tau '={}&\tfrac{3}{4} \Psi_{3}
 + \tfrac{1}{16} (3 \upsilon_{2} - 6 \gamma_{21'} - 4 \eta_{21'}) \kappa_{1}^{-1}
 -  K_{\rho}{} \kappa '
 + \kappa ' \rho
 -  \sigma ' \tau
 -  \rho ' \bar{\tau}
 -  K_{\rho '}{} \tau '
 + \rho ' \tau ',\\
\edt \tau '={}&\tfrac{1}{2} \Psi_{2}
 -  K_{\rho}{} \rho '
 + \rho \rho '
 -  \bar{\rho} \rho '
 -  \sigma \sigma '
 -  K_{\tau}{} \tau '
 + \tau \tau '\nonumber\\*
& -  \tfrac{1}{8} \kappa_{1}^{-1} (\bar{\upsilon}_{1'} + 6 \gamma_{11'} + 2 \eta_{11'} + 4 \bar\Psi_{2} \bar{\kappa}_{1'} + 4 \bar{\chi}_{1'}),\\
\edtp \tau '={}&- \tfrac{1}{8} (3 \gamma_{20'} + 2 \eta_{20'}) \kappa_{1}^{-1}
 -  \rho ' \bar{\sigma}
 -  K_{\rho}{} \sigma '
 -  K_{\tau '}{} \tau '
 + \tau '^2.\label{eq:Ethptau}
\end{align}
\end{subequations}
\end{lemma}
\begin{proof}
A direct translation of \eqref{eq:Twistkappa1} to GHP gives \eqref{eq:Thokappa1}-\eqref{eq:Ethpkappa1}, while \eqref{eq:TwistPsi2} gives \eqref{eq:ThoPsi2}-\eqref{eq:EthpPsi2}. Due to the relations \eqref{eq:xiGHP}, the GHP form of \eqref{eq:Divxi}, \eqref{eq:Curlxi}, \eqref{eq:CurlDgxi}, \eqref{eq:Twistxi} gives expressions for all GHP derivatives of $\rho ,\rho ',\tau ,\tau '$. However, these equations will also contain first order derivatives of $K_{\rho }{}, K_{\tau }{}, K_{\rho '}{}, K_{\tau '}{}$. All such derivatives can be eliminated by \eqref{eq:GHPPropagation1a}-\eqref{eq:GHPPropagation1d}, \eqref{eq:gammaGHP}, \eqref{eq:etaGHP}. This gives the relations \eqref{eq:Thorho}-\eqref{eq:Ethptau}. 
\end{proof}
\begin{remark}
\begin{itemize}
\item The $ \eta$ variables could have been avoided here as in Lemma~\ref{lem:DerNonsmallCovariant} if we would have used the components of $\xi$ instead of $\{\rho ,\rho ',\tau ,\tau '\}$. They are related via \eqref{eq:xiGHP}. As the spin coefficients are more commonly used in applications, the current form was chosen.
\item This lemma is important because all covariantly defined operators on Kerr expressed in GHP notation will have coefficients depending on the $\badset$ variables. In our setting, one can use the same GHP expressions for the operators on the perturbed spacetime. Lemma~\ref{lem:AllGHPDerNonSmall} allows for computation of commutators between such operators.
\end{itemize}
\end{remark}

\section{Space spinor formulation}
In this section, we will study the space spinor formulation of the evolution equations of section~\ref{sec:ApproxKillingCovariant} using the notation introduced in  
section~\ref{sec:SpaceSpinorsNotation}.

Define the space spinors 
\begin{subequations}
\begin{align}
\underline{H}={}&\tau \underset{3,1}{\overset{0,1}{\odot}}H,&
\underline{\xi}={}&\tfrac{2}{3}\tau \underset{1,1}{\overset{0,1}{\odot}}\xi ,&
\underline{I\xi}={}&\tfrac{1}{2}\underline{\xi}
 + \tfrac{1}{2}\widehat{\underline{\xi}},\\
\underline{\slashed{H}}={}&\tfrac{3}{2}\tau \underset{3,1}{\overset{1,1}{\odot}}H,&
\underline{\slashed{\xi}}={}&\tau \underset{1,1}{\overset{1,1}{\odot}}\xi ,&
\underline{I\slashed{\xi}}={}&\tfrac{1}{2}\underline{\slashed{\xi}}
 -  \tfrac{1}{2}\overline{\underline{\slashed{\xi}}}.
\end{align}
\end{subequations}
We can reconstruct the spacetime variables $H$, $\xi$ and $I\xi$ via
\begin{align}
H={}&\tfrac{1}{2}\tau \underset{2,0}{\overset{0,0}{\odot}}\underline{\slashed{H}}
 -  \tau \underset{4,0}{\overset{1,0}{\odot}}\underline{H},&
\xi ={}&- \tfrac{3}{2}\tau \underset{2,0}{\overset{1,0}{\odot}}\underline{\xi}
 + \tfrac{1}{2}\tau \underset{0,0}{\overset{0,0}{\odot}}\underline{\slashed{\xi}},&
I\xi ={}&- \tfrac{3}{2}\tau \underset{2,0}{\overset{1,0}{\odot}}\underline{I\xi}
 + \tfrac{1}{2}\tau \underset{0,0}{\overset{0,0}{\odot}}\underline{I\slashed{\xi}}.
\end{align}
We can also express the space spinors in terms of the Sen connection and the normal derivative.
\begin{align}
\underline{\slashed{H}}={}&\nabla_{\tau}\kappa
 + \scurl_{2} \kappa,&
\underline{H}={}&\stwist_{2} \kappa,&
\underline{\slashed{\xi}}={}&\sdiv_{2} \kappa,&
\underline{\xi}={}& \tfrac{2}{3}\scurl_{2} \kappa 
 - \tfrac{1}{3}\nabla_{\tau}\kappa.
\end{align}
The variables $\chi$, $\psi$, $\Psi$, $\Psi_2$ and $\kappa_1$ don't have any primed indices, so we can directly treat them as space spinors.

Now, we can make a covariant space spinor formulation of the evolution system in Lemma~\ref{lem:SmallVarPropCovariant} and Corollary~\ref{cor:TotalEvolCovariant}.

The first four equations in Lemma~\ref{lem:SmallVarPropCovariant} can be written as
\begin{subequations}
\label{eq:NormalDerHIxi}
\begin{align}
\nabla_{\tau}\underline{\slashed{H}}={}&
  \scurl_{2} \underline{\slashed{H}}
 - 3\sdiv_{4} \underline{H}
- A\underset{2,0}{\overset{1,0}{\odot}}\underline{\slashed{H}}
 -  \slashed{\Omega}\underline{\slashed{H}}
 + 3A\underset{4,0}{\overset{2,0}{\odot}}\underline{H}
 + 3\upsilon ,\\
\nabla_{\tau}\underline{H}={}&
 - 2\scurl_{4} \underline{H}
 + \stwist_{2} \underline{\slashed{H}}
 - A\underset{2,0}{\overset{0,0}{\odot}}\underline{\slashed{H}}
 + 2A\underset{4,0}{\overset{1,0}{\odot}}\underline{H}
 -  \slashed{\Omega}\underline{H}
 - 4\kappa \underset{4,0}{\overset{1,0}{\odot}}\psi ,\\
\nabla_{\tau}\underline{I\slashed{\xi}}={}&
 - 3\sdiv_{2} \underline{I\xi}
 + 3A\underset{2,0}{\overset{2,0}{\odot}}\underline{I\xi}
 -  \slashed{\Omega}\underline{I\slashed{\xi}},\\
\nabla_{\tau}\underline{I\xi}={}&
 - 2\scurl_{2} \underline{I\xi}
+2A\underset{2,0}{\overset{1,0}{\odot}}\underline{I\xi}
 -  \slashed{\Omega}\underline{I\xi}
 -  \tfrac{2}{3} \underline{I\slashed{\xi}}A
 + \tfrac{2}{3}\stwist_{0} \underline{I\slashed{\xi}}
 + 2\chi .
\end{align}
\end{subequations}
Similar to the Maxwell equation, the equation \eqref{eq:CurlDgchi} splits into one constraint equation and one evolution equation
\begin{subequations}
\begin{align}
\sdiv_{2} \chi={}&- \tfrac{1}{4}\sdiv_{2} \upsilon
 -  \tfrac{1}{4}\sdiv_{2} \widehat{\upsilon}
 + \frac{3 \Psi_{2}}{4 \kappa_{1}{}^2}\kappa \underset{2,0}{\overset{2,0}{\odot}}\kappa \underset{4,0}{\overset{2,0}{\odot}}\underline{H}
 -  \frac{3 \bar\Psi_{2}}{4 \bar{\kappa}_{1'}{}^2}\widehat{\kappa}\underset{2,0}{\overset{2,0}{\odot}}\widehat{\kappa}\underset{4,0}{\overset{2,0}{\odot}}\widehat{\underline{H}}
+ \tfrac{1}{2}\psi \underset{4,0}{\overset{4,0}{\odot}}\underline{H}
 -  \tfrac{1}{2}\widehat{\psi}\underset{4,0}{\overset{4,0}{\odot}}\widehat{\underline{H}},\\
\nabla_{\tau}\chi ={}&2\scurl_{2} \chi
 -  \tfrac{1}{4}\nabla_{\tau}\upsilon
 + \tfrac{1}{2}\scurl_{2} \upsilon
 -  \tfrac{1}{4}\nabla_{\tau}\widehat{\upsilon}
 -  \tfrac{1}{2}\scurl_{2} \widehat{\upsilon}
 -  \tfrac{1}{3}\slashed{\Omega}\underset{2,0}{\overset{0,0}{\odot}}\widehat{\upsilon}
 + A\underset{2,0}{\overset{1,0}{\odot}}\widehat{\upsilon}
 + \tfrac{1}{2}\Omega \underset{2,0}{\overset{2,0}{\odot}}\widehat{\upsilon}
 + \tfrac{1}{2} \Psi_{2}\underline{\slashed{H}}
 + \tfrac{1}{2} \bar\Psi_{2}\widehat{\underline{\slashed{H}}}
 \nonumber\\
&
 + \frac{3 \Psi_{2}}{4 \kappa_{1}{}^2}\kappa \underset{0,0}{\overset{0,0}{\odot}}\kappa \underset{2,0}{\overset{2,0}{\odot}}\underline{\slashed{H}}
 + \frac{3 \bar\Psi_{2}}{4 \bar{\kappa}_{1'}{}^2}\widehat{\kappa}\underset{0,0}{\overset{0,0}{\odot}}\widehat{\kappa}\underset{2,0}{\overset{2,0}{\odot}}\widehat{\underline{\slashed{H}}}
 + \frac{3 \Psi_{2}}{2 \kappa_{1}{}^2}\kappa \underset{2,0}{\overset{1,0}{\odot}}\kappa \underset{4,0}{\overset{2,0}{\odot}}\underline{H}
 + \frac{3 \bar\Psi_{2}}{2 \bar{\kappa}_{1'}{}^2}\widehat{\kappa}\underset{2,0}{\overset{1,0}{\odot}}\widehat{\kappa}\underset{4,0}{\overset{2,0}{\odot}}\widehat{\underline{H}}
\nonumber\\
&
  + \tfrac{1}{2}\psi \underset{2,0}{\overset{2,0}{\odot}}\underline{\slashed{H}}
 + \tfrac{1}{2}\widehat{\psi}\underset{2,0}{\overset{2,0}{\odot}}\widehat{\underline{\slashed{H}}}
 + \psi \underset{4,0}{\overset{3,0}{\odot}}\underline{H}
 + \widehat{\psi}\underset{4,0}{\overset{3,0}{\odot}}\widehat{\underline{H}}. \label{eq:NormalDerchi}
\end{align}
\end{subequations}
The space spinor split of the full vacuum Bianchi system \eqref{eq:VacuumBianchiCovariant} gives a constraint equation and an evolution equation
\begin{subequations}
\begin{align}
\sdiv_{4} \Psi={}&0,\\
\nabla_{\tau}\Psi ={}&2\scurl_{4} \Psi . \label{eq:NormalFullPsi}
\end{align}
\end{subequations}
Observe that this system is equivalent to both equations \eqref{eq:TwistPsi2Alt} and \eqref{eq:CurlDgpsi} together.
The space spinor split of equations \eqref{eq:Twistkappa1}, \eqref{eq:Divxi}, \eqref{eq:Curlxi} are
\begin{subequations}
\begin{align}
\nabla_{\tau}\kappa_{1}{}={}&- \frac{1}{3 \kappa_{1}{}}\kappa \underset{2,0}{\overset{2,0}{\odot}}\underline{\slashed{H}}
 + \frac{1}{2 \kappa_{1}{}}\kappa \underset{2,0}{\overset{2,0}{\odot}}\underline{\xi},\label{eq:NormalDerkappa1}\\
\stwist_{0} \kappa_{1}{}={}&- \frac{1}{6 \kappa_{1}{}}\kappa \underset{2,0}{\overset{1,0}{\odot}}\underline{\slashed{H}}
 -  \frac{1}{2 \kappa_{1}{}}\kappa \underset{4,0}{\overset{2,0}{\odot}}\underline{H}
 -  \frac{1}{2 \kappa_{1}{}}\kappa \underset{2,0}{\overset{1,0}{\odot}}\underline{\xi}
 -  \frac{1}{6 \kappa_{1}{}}\kappa \underset{0,0}{\overset{0,0}{\odot}}\underline{\slashed{\xi}},\\
\nabla_{\tau}\underline{\slashed{\xi}}={}&
 - 3\sdiv_{2} \underline{\xi}
 + 3A\underset{2,0}{\overset{2,0}{\odot}}\underline{\xi}
 -  \underline{\slashed{\xi}}\underset{0,0}{\overset{0,0}{\odot}}\slashed{\Omega},\label{eq:NormalDerxi0}\\
\nabla_{\tau}\underline{\xi}={}&
 - 2\scurl_{2} \underline{\xi}
 + \tfrac{2}{3}\stwist_{0} \underline{\slashed{\xi}}
 + 2A\underset{2,0}{\overset{1,0}{\odot}}\underline{\xi}
 -  \slashed{\Omega}\underline{\xi}
 -  \tfrac{2}{3}A\underset{0,0}{\overset{0,0}{\odot}}\underline{\slashed{\xi}}
 -  \upsilon
 + 2\kappa \underset{4,0}{\overset{2,0}{\odot}}\psi
 - 4 \Psi_{2}\kappa . \label{eq:NormalDerxi2}
\end{align}
\end{subequations}
\begin{lemma} 
The space spinor form of the evolution system in Corollary~\ref{cor:TotalEvol} is given by \eqref{eq:NormalDerHIxi}, \eqref{eq:NormalDerchi},\eqref{eq:NormalFullPsi}, \eqref{eq:NormalDerkappa1}, \eqref{eq:NormalDerxi0}, \eqref{eq:NormalDerxi2}.
Under Assumption~\ref{ass:NiceGaugeSourceFunction} this system is first order symmetric hyperbolic.
\end{lemma}
\begin{proof}
Under Assumption~\ref{ass:NiceGaugeSourceFunction}, we can interpret all $\upsilon$ terms as lower order. 
The system given by \eqref{eq:NormalDerHIxi}, \eqref{eq:NormalDerchi},\eqref{eq:NormalFullPsi}, \eqref{eq:NormalDerkappa1}, \eqref{eq:NormalDerxi0}, \eqref{eq:NormalDerxi2} has the form
\begin{align}
\begin{bmatrix}
\nabla_{\tau}\underline{\slashed{H}}\\
\sqrt{3}\nabla_{\tau}\underline{H}\\
\sqrt{2}\nabla_{\tau}\underline{I\slashed{\xi}}\\
3\nabla_{\tau}\underline{I\xi}\\
\nabla_{\tau}\chi \\
\nabla_{\tau}\Psi \\
\nabla_{\tau}\kappa_{1}{}\\
\sqrt{2}\nabla_{\tau}\underline{\slashed{\xi}}\\
3\nabla_{\tau}\underline{\xi}
\end{bmatrix}={}&\begin{bmatrix}
\scurl_{2} & - \sqrt{3}\sdiv_{4} & 0 & 0 & 0 & 0 & 0 & 0 & 0\\
\sqrt{3}\stwist_{2} & -2\scurl_{4} & 0 & 0 & 0 & 0 & 0 & 0 & 0\\
0 & 0 & 0 & - \sqrt{2}\sdiv_{2} & 0 & 0 & 0 & 0 & 0\\
0 & 0 & \sqrt{2}\stwist_{0} & -2\scurl_{2} & 0 & 0 & 0 & 0 & 0\\
0 & 0 & 0 & 0 & 2\scurl_{2} & 0 & 0 & 0 & 0\\
0 & 0 & 0 & 0 & 0 & 2\scurl_{4} & 0 & 0 & 0\\
0 & 0 & 0 & 0 & 0 & 0 & 0 & 0 & 0\\
0 & 0 & 0 & 0 & 0 & 0 & 0 & 0 & - \sqrt{2}\sdiv_{2}\\
0 & 0 & 0 & 0 & 0 & 0 & 0 & \sqrt{2}\stwist_{0} & -2\scurl_{2}
\end{bmatrix}\begin{bmatrix}
\underline{\slashed{H}}\\
\sqrt{3}\underline{H}\\
\sqrt{2}\underline{I\slashed{\xi}}\\
3\underline{I\xi}\\
\chi \\
\Psi \\
\kappa_{1}{}\\
\sqrt{2}\underline{\slashed{\xi}}\\
3\underline{\xi}
\end{bmatrix}+\text{l.o.}
\end{align}
The matrix is anti-self adjoint up to lower order terms with respect to the Hermitian adjoint. Hence, the system is  first order symmetric hyperbolic.
\end{proof}

We can also study the space spinor version of the symmetric hyperbolic system in Corollary~\ref{cor:SmallEvol}. We still have the evolution equations \eqref{eq:NormalDerHIxi} and \eqref{eq:NormalDerchi}. The equation \eqref{eq:NormalFullPsi} from the Bianchi system can no longer be used because it contains equations for the non-small variable $\Psi_2$. Therefore, we need the space spinor version of equation \eqref{eq:CurlDgpsi}. After splitting it into different spin components, we find that it is equivalent to
\begin{subequations}
\begin{align}
\mathcal{P}^{2}_{4,0} \nabla_{\tau}\psi ={}&2\mathcal{P}^{2}_{4,0} \scurl_{4} \psi
 -  \frac{3 \Psi_{2}}{\kappa_{1}^2}\mathcal{P}^{2}_{4,0} \kappa \underset{4,0}{\overset{1,0}{\odot}}\underline{H},\\
\mathcal{P}^{1}_{4,0} \nabla_{\tau}\psi ={}&2\mathcal{P}^{1}_{4,0} \scurl_{4} \psi
 + \frac{3}{\kappa_{1}^2}\mathcal{P}^{1}_{4,0} \kappa \underset{2,0}{\overset{0,0}{\odot}}\kappa \underset{2,0}{\overset{1,0}{\odot}}\sdiv_{4} \psi
 -  \frac{3 \Psi_{2}}{\kappa_{1}^2}\mathcal{P}^{1}_{4,0} \kappa \underset{2,0}{\overset{0,0}{\odot}}\underline{\slashed{H}}
 -  \frac{6 \Psi_{2}}{\kappa_{1}^2}\mathcal{P}^{1}_{4,0} \kappa \underset{4,0}{\overset{1,0}{\odot}}\underline{H}.
\end{align}
\end{subequations}
It clearly gives a time evolution, but the interpretation of these equations as first order symmetric hyperbolic is complicated by the fact that a-priori, the projection operators will not commute with Hermitian conjugation. In general the Hermitian conjugate will mix GHP components unless the dyad is adapted to the normal $\tau^{AA'}$, but here, we have adapted the dyad to $\kappa_{AB}$ instead. Projections are needed to describe that the middle component of $\psi$ is zero. Therefore, we stick to the GHP component version in Corollary~\ref{cor:TotalEvol}.

Here, we have treated the foliation and the normal $\tau$ as a-priori given and therefore we have not included the second fundamental form in the evolution system. Alternatively, one can specify the acceleration $A$ and use the evolution equations for the extrinsic curvature
\begin{subequations}
\begin{align}
\nabla_{\tau}\slashed{\Omega}={}&2\sdiv_{2} A
 -  \tfrac{1}{3} \slashed{\Omega}^2
 - 2A\underset{2,0}{\overset{2,0}{\odot}}A
 -  \Omega \underset{4,0}{\overset{4,0}{\odot}}\Omega ,\\
\nabla_{\tau}\Omega ={}&
 - 2\scurl_{4} \Omega
 + 2\stwist_{2} A
 -  \tfrac{5}{3} \slashed{\Omega}\,\Omega
 - 2A\underset{2,0}{\overset{0,0}{\odot}}A
 + 6A\underset{4,0}{\overset{1,0}{\odot}}\Omega
 + 2\Omega \underset{4,0}{\overset{2,0}{\odot}}\Omega
 + 2\Psi.
\end{align}
\end{subequations}
Observe that we also have the constraint equation
\begin{align}
\sdiv_{4} \Omega={}&\tfrac{2}{3}\stwist_{0} \slashed{\Omega}.
\end{align}

All the equations above in this section holds on any slice of the foliation. On the initial surface, the equations \eqref{eq:InitialElliptic} and \eqref{eq:InitialNormal} take the form
\begin{align}
\sdiv_{4} \underline{H} - 2\Omega \underset{4,0}{\overset{3,0}{\odot}}\underline{H}={}&0,&
\underline{\slashed{H}}={}&0.
\end{align}

\section{Killing vectors}
To prove estimates, or other detailed analysis of the perturbed spacetime, we need good approximations of as many geometrically defined structures of the Kerr spacetime as we can. We have already seen that we can find good approximations of the Killing spinor $\kappa$, a principal dyad, the corresponding spin coefficients and curvature components. Now, we can focus on the Killing vectors and a radial coordinate.

Due to the relations \eqref{eq:Divxi} and \eqref{eq:Twistxi}, we see that the real part of $\xi$ is an approximation of the asymptotically timelike Killing vector.
The vector
\begin{align}
\zeta ={}&- \tfrac{9}{4} (\kappa_{1}^2 + \bar{\kappa}_{1'}^2)\xi
 + \tfrac{9}{2}\bar{\kappa}\underset{1,1}{\overset{0,1}{\odot}}\kappa \underset{1,1}{\overset{1,0}{\odot}}\xi ,
\label{eq:zetaDef}
\end{align}
is also an approximate Killing vector. As a direct consequence of Lemma~\ref{lem:DerNonsmallCovariant} we have the following lemma.
\begin{lemma}\label{lem:zetaDer}
The vector $\zeta$ satisfies
\begin{subequations}
\begin{align}
\sDiv_{1,1} \zeta ={}&\tfrac{9}{4}\xi \underset{1,1}{\overset{1,1}{\odot}}\kappa \underset{3,1}{\overset{2,0}{\odot}}H
 + \tfrac{9}{4}\xi \underset{1,1}{\overset{1,1}{\odot}}\bar{\kappa}\underset{1,3}{\overset{0,2}{\odot}}\bar{H}
 - 3\xi \underset{1,1}{\overset{1,1}{\odot}}\bar{\kappa}\underset{1,1}{\overset{0,1}{\odot}}I\xi
 + 9\xi \underset{1,1}{\overset{1,1}{\odot}}\kappa \underset{1,1}{\overset{1,0}{\odot}}I\xi
 + \tfrac{27}{4}\bar{\kappa}\underset{0,2}{\overset{0,2}{\odot}}\kappa \underset{2,2}{\overset{2,0}{\odot}}\gamma  \label{eq:DivZetaEq1a},\\
\sCurl_{1,1} \zeta ={}&\tfrac{27}{16} (\kappa_{1}^2 + \bar{\kappa}_{1'}^2)\upsilon
 -  \tfrac{27}{8} (\kappa_{1}^2 + \bar{\kappa}_{1'}^2)\kappa \underset{4,0}{\overset{2,0}{\odot}}\psi
 + \bigl(\tfrac{27}{4} \Psi_{2} (\kappa_{1}^2 + \bar{\kappa}_{1'}^2) -  \tfrac{27}{2} \bar\Psi_{2} \bar{\kappa}_{1'}^2\bigr)\kappa\nonumber\\
& -  \tfrac{9}{4}\kappa \underset{4,0}{\overset{2,0}{\odot}}\xi \underset{3,1}{\overset{0,1}{\odot}}H
 -  \tfrac{9}{8}\kappa \underset{2,0}{\overset{1,0}{\odot}}\xi \underset{3,1}{\overset{1,1}{\odot}}H
 + \tfrac{3}{2}\kappa \underset{0,0}{\overset{0,0}{\odot}}\xi \underset{1,1}{\overset{1,1}{\odot}}\xi
 -  \tfrac{9}{4}\bar{\kappa}\underset{2,2}{\overset{0,2}{\odot}}\xi \underset{1,3}{\overset{0,1}{\odot}}\bar{H}
 - 3\bar{\kappa}\underset{2,2}{\overset{0,2}{\odot}}\xi \underset{1,1}{\overset{0,0}{\odot}}I\xi\nonumber\\
& + 3\bar{\kappa}\underset{2,2}{\overset{0,2}{\odot}}\xi \underset{1,1}{\overset{0,0}{\odot}}\xi
 + 9\kappa \underset{2,0}{\overset{1,0}{\odot}}\xi \underset{1,1}{\overset{0,1}{\odot}}I\xi
 -  \tfrac{9}{2}\kappa \underset{0,0}{\overset{0,0}{\odot}}\xi \underset{1,1}{\overset{1,1}{\odot}}I\xi
 + \tfrac{9}{2}\bar{\kappa}\underset{2,2}{\overset{0,2}{\odot}}\xi \underset{3,1}{\overset{1,0}{\odot}}H
 -  \tfrac{27}{4}\kappa \underset{2,0}{\overset{1,0}{\odot}}\bar{\kappa}\underset{2,2}{\overset{0,2}{\odot}}\gamma\nonumber\\
& + \tfrac{27}{4}\kappa \underset{0,0}{\overset{0,0}{\odot}}\bar{\kappa}\underset{0,2}{\overset{0,2}{\odot}}\bar{\chi}
 + \tfrac{27}{16}\kappa \underset{0,0}{\overset{0,0}{\odot}}\bar{\kappa}\underset{0,2}{\overset{0,2}{\odot}}\bar{\upsilon} \label{eq:CurlZetaEq1b},\\
\sCurlDagger_{1,1} \zeta ={}&\tfrac{27}{4} (\kappa_{1}^2 + \bar{\kappa}_{1'}^2)\bar{\chi}
 + \tfrac{27}{16} (\kappa_{1}^2 + \bar{\kappa}_{1'}^2)\bar{\upsilon}
 -  \tfrac{27}{8} (\kappa_{1}^2 + \bar{\kappa}_{1'}^2)\bar{\kappa}\underset{0,4}{\overset{0,2}{\odot}}\bar{\psi}\nonumber\\
& + \bigl(\tfrac{27}{4} \bar\Psi_{2} (\kappa_{1}^2 + \bar{\kappa}_{1'}^2) -  \tfrac{27}{2} \Psi_{2} \kappa_{1}^2\bigr)\bar{\kappa}
 -  \tfrac{9}{4}\kappa \underset{2,2}{\overset{2,0}{\odot}}\xi \underset{3,1}{\overset{1,0}{\odot}}H
 + 3\kappa \underset{2,2}{\overset{2,0}{\odot}}\xi \underset{1,1}{\overset{0,0}{\odot}}\xi
 -  \tfrac{9}{4}\bar{\kappa}\underset{0,4}{\overset{0,2}{\odot}}\xi \underset{1,3}{\overset{1,0}{\odot}}\bar{H}\nonumber\\
& -  \tfrac{9}{8}\bar{\kappa}\underset{0,2}{\overset{0,1}{\odot}}\xi \underset{1,3}{\overset{1,1}{\odot}}\bar{H}
 + 3\bar{\kappa}\underset{0,2}{\overset{0,1}{\odot}}\xi \underset{1,1}{\overset{1,0}{\odot}}I\xi
 + \tfrac{3}{2}\bar{\kappa}\underset{0,0}{\overset{0,0}{\odot}}\xi \underset{1,1}{\overset{1,1}{\odot}}I\xi
 + \tfrac{3}{2}\bar{\kappa}\underset{0,0}{\overset{0,0}{\odot}}\xi \underset{1,1}{\overset{1,1}{\odot}}\xi
 + \tfrac{9}{2}\kappa \underset{2,2}{\overset{2,0}{\odot}}\xi \underset{1,3}{\overset{0,1}{\odot}}\bar{H}\nonumber\\
& - 3\kappa \underset{2,2}{\overset{2,0}{\odot}}\xi \underset{1,1}{\overset{0,0}{\odot}}I\xi
 -  \tfrac{27}{4}\bar{\kappa}\underset{0,2}{\overset{0,1}{\odot}}\kappa \underset{2,2}{\overset{2,0}{\odot}}\gamma
 + \tfrac{27}{16}\bar{\kappa}\underset{0,0}{\overset{0,0}{\odot}}\kappa \underset{2,0}{\overset{2,0}{\odot}}\upsilon  \label{eq:CurlDgZetaEq1a},\\
\sTwist_{1,1} \zeta ={}&- \tfrac{27}{8} (\kappa_{1}^2 + \bar{\kappa}_{1'}^2)\gamma
 + \tfrac{9}{4}\xi \underset{1,1}{\overset{0,0}{\odot}}\kappa \underset{3,1}{\overset{2,0}{\odot}}H
 + \tfrac{9}{4}\xi \underset{1,1}{\overset{0,0}{\odot}}\bar{\kappa}\underset{1,3}{\overset{0,2}{\odot}}\bar{H}
 - 3\xi \underset{1,1}{\overset{0,0}{\odot}}\bar{\kappa}\underset{1,1}{\overset{0,1}{\odot}}I\xi
 + \tfrac{9}{2}\xi \underset{3,3}{\overset{1,1}{\odot}}\kappa \underset{1,3}{\overset{0,0}{\odot}}\bar{H}\nonumber\\
& -  \tfrac{3}{2}\xi \underset{1,3}{\overset{0,1}{\odot}}\kappa \underset{1,3}{\overset{1,0}{\odot}}\bar{H}
 - 3\xi \underset{3,1}{\overset{1,0}{\odot}}\kappa \underset{1,1}{\overset{0,0}{\odot}}I\xi
 + \xi \underset{1,1}{\overset{0,0}{\odot}}\kappa \underset{1,1}{\overset{1,0}{\odot}}I\xi
 + \tfrac{9}{2}\xi \underset{3,3}{\overset{1,1}{\odot}}\bar{\kappa}\underset{3,1}{\overset{0,0}{\odot}}H
 -  \tfrac{3}{2}\xi \underset{3,1}{\overset{1,0}{\odot}}\bar{\kappa}\underset{3,1}{\overset{0,1}{\odot}}H\nonumber\\
& + \tfrac{27}{4}\bar{\kappa}\underset{2,2}{\overset{0,1}{\odot}}\kappa \underset{2,2}{\overset{1,0}{\odot}}\gamma
 -  \tfrac{27}{4}\bar{\kappa}\underset{2,2}{\overset{0,1}{\odot}}\kappa \underset{0,2}{\overset{0,0}{\odot}}\bar{\chi}
 -  \tfrac{27}{16}\bar{\kappa}\underset{2,2}{\overset{0,1}{\odot}}\kappa \underset{0,2}{\overset{0,0}{\odot}}\bar{\upsilon}
 + \tfrac{27}{8}\bar{\kappa}\underset{2,2}{\overset{0,1}{\odot}}\bar{\kappa}\underset{2,4}{\overset{0,2}{\odot}}\kappa \underset{0,4}{\overset{0,0}{\odot}}\bar{\psi}
 -  \tfrac{27}{16}\bar{\kappa}\underset{2,0}{\overset{0,0}{\odot}}\kappa \underset{2,0}{\overset{1,0}{\odot}}\upsilon\nonumber\\
& + \tfrac{27}{8}\bar{\kappa}\underset{2,0}{\overset{0,0}{\odot}}\kappa \underset{2,0}{\overset{1,0}{\odot}}\kappa \underset{4,0}{\overset{2,0}{\odot}}\psi  \label{eq:TwistZetaEq1a}.
\end{align}
\end{subequations}
\end{lemma}
Note that the right hand sides of \eqref{eq:DivZetaEq1a} and \eqref{eq:TwistZetaEq1a} are linear in the small variables $H$, $I\xi$, $\chi$, $\upsilon$, $\psi$, $\gamma$, $\Upsilon$. Hence, $\zeta$ is a real Killing vector on Kerr, and an approximate Killing vector in general.
In Boyer-Lindquist coordinates on Kerr $\zeta^{a}= a^2 (\partial_{t})^{a} + a (\partial_{\phi})^{a}$.

In general $\xi$ and $\zeta$ are not real, but the real parts will give good approximations of the two Killing vectors in Kerr. Observe that in Schwarzschild $\zeta=0$, so we can only extract one Killing vector in that case.

In GHP notation, the components of $\zeta$ are
\begin{subequations}
\begin{align}
\zeta_{00'}={}&\tfrac{9}{4} \kappa_{1} (\kappa_{1} -  \bar{\kappa}_{1'})^2 (K_{\rho}{} - 3 \rho),&
\zeta_{01'}={}&\tfrac{9}{4} \kappa_{1} (\kappa_{1} + \bar{\kappa}_{1'})^2 (K_{\tau}{} - 3 \tau),\\
\zeta_{10'}={}&- \tfrac{9}{4} \kappa_{1} (\kappa_{1} + \bar{\kappa}_{1'})^2 (K_{\tau '}{} - 3 \tau '),&
\zeta_{11'}={}&- \tfrac{9}{4} \kappa_{1} (\kappa_{1} -  \bar{\kappa}_{1'})^2 (K_{\rho '}{} - 3 \rho ').
\end{align}
\end{subequations}

In Kerr the Boyer-Lindquist radial coordinate can be reconstructed from $\kappa_1$ via 
\begin{align}
r={}&- \tfrac{3}{2} (\kappa_{1} + \bar{\kappa}_{1'}).\label{eq:rDef}
\end{align}
We can use this as a radial coordinate also on the perturbed spacetime. That this coordinate has all the properties needed for the applications will have to be established before it is used.

\section{Approximate constants}
In many cases, good approximations of the mass and angular momentum parameters are needed. Such approximations are often difficult to find due to the non-local nature of the mass and angular momentum. Here, we find local expressions in terms of our non-small variables $\kappa_{1}$, $\Psi_{2}$, $\xi$ and $\zeta$. Observe however that solving \eqref{eq:InitialElliptic} on the initial surface is a non-local operation.
\begin{definition}
Given an approximate Killing spinor $\kappa$ and the corresponding $\xi$ given by \eqref{eq:xiDef} and $\zeta$ given by \eqref{eq:zetaDef}, we define the \emph{approximate constants}
\begin{subequations}
\begin{align}
\mathit{c}_{M}{}={}&\tfrac{27}{2} \Psi_{2} \kappa_{1}^3
 + \tfrac{27}{2} \bar\Psi_{2} \bar{\kappa}_{1'}^3,\\
\mathit{c}_{N}{}={}&\tfrac{27}{2} \Psi_{2} \kappa_{1}^3
 -  \tfrac{27}{2} \bar\Psi_{2} \bar{\kappa}_{1'}^3,\\
\mathit{c}_{1}{}={}&\xi \underset{1,1}{\overset{1,1}{\odot}}\bar{\xi}
 - 9 \kappa_{1}^2\Psi_{2}
 - 9 \bar{\kappa}_{1'}^2\bar\Psi_{2},\\
\mathit{c}_{a^{2}{}}{}={}&\tfrac{1}{2}\xi \underset{1,1}{\overset{1,1}{\odot}}\zeta
 + \tfrac{1}{2}\bar{\xi}\underset{1,1}{\overset{1,1}{\odot}}\bar{\zeta}
 -  \tfrac{81}{4} (\kappa_{1}^2 -  \bar{\kappa}_{1'}^2)\kappa_{1}^2\Psi_{2}
 + \tfrac{81}{4} (\kappa_{1}^2 -  \bar{\kappa}_{1'}^2) \bar{\kappa}_{1'}^2\bar\Psi_{2} .
\end{align}
\end{subequations}
\end{definition}

\begin{lemma}\label{lem:GradientApproxConst}
The approximate constants satisfy
\begin{subequations}
\begin{align}
\sTwist_{0,0} \mathit{c}_{M}{}={}&\tfrac{27}{2} \kappa_{1}^2\Upsilon
 + \tfrac{27}{2} \bar{\kappa}_{1'}^2\bar{\Upsilon},\label{eq:TwistcM}\\
\sTwist_{0,0} \mathit{c}_{N}{}={}&\tfrac{27}{2} \kappa_{1}^2\Upsilon
 -  \tfrac{27}{2} \bar{\kappa}_{1'}^2\bar{\Upsilon},\label{eq:TwistcN}\\
\sTwist_{0,0} \mathit{c}_{1}{}={}&- \tfrac{3}{2}\xi \underset{2,0}{\overset{1,0}{\odot}}\chi
 -  \tfrac{3}{2}\bar{\xi}\underset{0,2}{\overset{0,1}{\odot}}\bar{\chi}
 -  \tfrac{3}{8}\xi \underset{0,2}{\overset{0,1}{\odot}}\bar{\upsilon}
 -  \tfrac{3}{8}\xi \underset{2,0}{\overset{1,0}{\odot}}\upsilon
 -  \tfrac{3}{8}\bar{\xi}\underset{0,2}{\overset{0,1}{\odot}}\bar{\upsilon}
 -  \tfrac{3}{8}\bar{\xi}\underset{2,0}{\overset{1,0}{\odot}}\upsilon
 + \tfrac{3}{4}\xi \underset{0,2}{\overset{0,1}{\odot}}\bar{\kappa}\underset{0,4}{\overset{0,2}{\odot}}\bar{\psi}\nonumber\\
& + \tfrac{3}{4}\xi \underset{2,0}{\overset{1,0}{\odot}}\kappa \underset{4,0}{\overset{2,0}{\odot}}\psi
 + \tfrac{3}{4}\bar{\xi}\underset{0,2}{\overset{0,1}{\odot}}\bar{\kappa}\underset{0,4}{\overset{0,2}{\odot}}\bar{\psi}
 + \tfrac{3}{4}\bar{\xi}\underset{2,0}{\overset{1,0}{\odot}}\kappa \underset{4,0}{\overset{2,0}{\odot}}\psi
 + \tfrac{3}{2}\xi \underset{2,2}{\overset{1,1}{\odot}}\bar{\gamma}
 + \tfrac{3}{2}\bar{\xi}\underset{2,2}{\overset{1,1}{\odot}}\gamma
 -  \tfrac{9}{2} \Psi_{2}\kappa \underset{3,1}{\overset{2,0}{\odot}}H\nonumber\\
& - 3 \Psi_{2}\kappa \underset{1,1}{\overset{1,0}{\odot}}I\xi
 -  \tfrac{9}{2} \bar\Psi_{2}\bar{\kappa}\underset{1,3}{\overset{0,2}{\odot}}\bar{H}
 + 3 \bar\Psi_{2}\bar{\kappa}\underset{1,1}{\overset{0,1}{\odot}}I\xi
 - 9 \kappa_{1}\Upsilon
 - 9 \bar{\kappa}_{1'}\bar{\Upsilon},\label{eq:Twistc1}\\
\sTwist_{0,0} \mathit{c}_{a^{2}{}}{}={}&-3I\xi \underset{2,0}{\overset{1,0}{\odot}}\bar{\kappa}\underset{2,2}{\overset{0,2}{\odot}}\bar{\xi}\underset{1,1}{\overset{0,0}{\odot}}\bar{\xi}
 -  \tfrac{3}{2}I\xi \underset{2,0}{\overset{1,0}{\odot}}\kappa \underset{0,0}{\overset{0,0}{\odot}}\bar{\xi}\underset{1,1}{\overset{1,1}{\odot}}\bar{\xi}
 -  \tfrac{3}{2}\zeta \underset{0,2}{\overset{0,1}{\odot}}\bar{\chi}
 -  \tfrac{3}{2}\bar{\zeta}\underset{2,0}{\overset{1,0}{\odot}}\chi
 -  \tfrac{3}{8}\zeta \underset{0,2}{\overset{0,1}{\odot}}\bar{\upsilon}
 -  \tfrac{3}{8}\zeta \underset{2,0}{\overset{1,0}{\odot}}\upsilon\nonumber\\
& -  \tfrac{3}{8}\bar{\zeta}\underset{0,2}{\overset{0,1}{\odot}}\bar{\upsilon}
 -  \tfrac{3}{8}\bar{\zeta}\underset{2,0}{\overset{1,0}{\odot}}\upsilon
 + \tfrac{3}{4}\zeta \underset{0,2}{\overset{0,1}{\odot}}\bar{\kappa}\underset{0,4}{\overset{0,2}{\odot}}\bar{\psi}
 + \tfrac{3}{4}\zeta \underset{2,0}{\overset{1,0}{\odot}}\kappa \underset{4,0}{\overset{2,0}{\odot}}\psi
 + \tfrac{3}{4}\bar{\zeta}\underset{0,2}{\overset{0,1}{\odot}}\bar{\kappa}\underset{0,4}{\overset{0,2}{\odot}}\bar{\psi}
 + \tfrac{3}{4}\bar{\zeta}\underset{2,0}{\overset{1,0}{\odot}}\kappa \underset{4,0}{\overset{2,0}{\odot}}\psi\nonumber\\
& + \tfrac{9}{8}H\underset{2,0}{\overset{2,0}{\odot}}\kappa \underset{0,0}{\overset{0,0}{\odot}}\xi \underset{1,1}{\overset{1,1}{\odot}}\xi
 + \tfrac{9}{8}H\underset{2,0}{\overset{2,0}{\odot}}\kappa \underset{0,0}{\overset{0,0}{\odot}}\bar{\xi}\underset{1,1}{\overset{1,1}{\odot}}\bar{\xi}
 + \tfrac{9}{8}\bar{H}\underset{0,2}{\overset{0,2}{\odot}}\bar{\kappa}\underset{0,0}{\overset{0,0}{\odot}}\xi \underset{1,1}{\overset{1,1}{\odot}}\xi
 + \tfrac{9}{8}\bar{H}\underset{0,2}{\overset{0,2}{\odot}}\bar{\kappa}\underset{0,0}{\overset{0,0}{\odot}}\bar{\xi}\underset{1,1}{\overset{1,1}{\odot}}\bar{\xi}\nonumber\\
& + \tfrac{3}{2}I\xi \underset{0,2}{\overset{0,1}{\odot}}\bar{\kappa}\underset{0,0}{\overset{0,0}{\odot}}\xi \underset{1,1}{\overset{1,1}{\odot}}\xi
 + \tfrac{3}{2}\zeta \underset{2,2}{\overset{1,1}{\odot}}\gamma
 + \tfrac{3}{2}\bar{\zeta}\underset{2,2}{\overset{1,1}{\odot}}\bar{\gamma}
 + \tfrac{9}{4}H\underset{2,0}{\overset{2,0}{\odot}}\bar{\kappa}\underset{2,2}{\overset{0,2}{\odot}}\xi \underset{1,1}{\overset{0,0}{\odot}}\xi
 + \tfrac{9}{4}H\underset{2,0}{\overset{2,0}{\odot}}\bar{\kappa}\underset{2,2}{\overset{0,2}{\odot}}\bar{\xi}\underset{1,1}{\overset{0,0}{\odot}}\bar{\xi}\nonumber\\
& + \tfrac{9}{4}\bar{H}\underset{0,2}{\overset{0,2}{\odot}}\kappa \underset{2,2}{\overset{2,0}{\odot}}\xi \underset{1,1}{\overset{0,0}{\odot}}\xi
 + \tfrac{9}{4}\bar{H}\underset{0,2}{\overset{0,2}{\odot}}\kappa \underset{2,2}{\overset{2,0}{\odot}}\bar{\xi}\underset{1,1}{\overset{0,0}{\odot}}\bar{\xi}
 + 3I\xi \underset{0,2}{\overset{0,1}{\odot}}\kappa \underset{2,2}{\overset{2,0}{\odot}}\xi \underset{1,1}{\overset{0,0}{\odot}}\xi
 + \tfrac{81}{4} (\kappa_{1}^2 -  \bar{\kappa}_{1'}^2) \bar{\kappa}_{1'}\bar{\Upsilon}\nonumber\\
& -  \tfrac{81}{4} (\kappa_{1}^3 -  \kappa_{1} \bar{\kappa}_{1'}^2)\Upsilon
 + \tfrac{27}{4} \bigl(\Psi_{2} (\kappa_{1}^2 + \bar{\kappa}_{1'}^2) - 2 \bar\Psi_{2} \bar{\kappa}_{1'}^2\bigr)\kappa \underset{1,1}{\overset{1,0}{\odot}}I\xi\nonumber\\
& + \tfrac{81}{8} \bigl(\Psi_{2} (\kappa_{1}^2 + \bar{\kappa}_{1'}^2) - 2 \bar\Psi_{2} \bar{\kappa}_{1'}^2\bigr)\kappa \underset{3,1}{\overset{2,0}{\odot}}H
 -  \tfrac{27}{4} \bigl(\bar\Psi_{2} (\kappa_{1}^2 + \bar{\kappa}_{1'}^2) - 2 \Psi_{2} \kappa_{1}^2\bigr)\bar{\kappa}\underset{1,1}{\overset{0,1}{\odot}}I\xi\nonumber\\
& + \tfrac{81}{8} \bigl(\bar\Psi_{2} (\kappa_{1}^2 + \bar{\kappa}_{1'}^2) - 2 \Psi_{2} \kappa_{1}^2\bigr)\bar{\kappa}\underset{1,3}{\overset{0,2}{\odot}}\bar{H}.\label{eq:TwistcA}
\end{align}
\end{subequations}
\end{lemma}
\begin{remark}\phantom{a}\\\vspace{-2ex}
\begin{itemize}
\item Note that the right hand sides are linear in the small variables $H, I\xi, \chi, \upsilon, \psi, \gamma, \Upsilon$. Hence, $\mathit{c}_{M}, \mathit{c}_{N}, \mathit{c}_{1}, \mathit{c}_{a^2}$ are constants on Kerr.
\item Now, when we have a full set of equations for all derivatives of the approximate constants, we can work with operators with coefficients depending on $\mathit{c}_{M}$, $\mathit{c}_{N}$, $\mathit{c}_{1}$ and $\mathit{c}_{a^{2}{}}$. We can also allow to rescale variables with such coefficients.
\end{itemize}
\end{remark}
\begin{proof}
The equations \eqref{eq:TwistcM}-\eqref{eq:Twistc1} follows from Lemma~\ref{lem:DerNonsmallCovariant} and the definition of $I\xi$. For \eqref{eq:TwistcA} also Lemma~\ref{lem:zetaDer} is need.
\end{proof}

In GHP notation, the approximate constants $\mathit{c}_{1}$ and $\mathit{c}_{a^{2}{}}$ are
\begin{subequations}
\begin{align}
\mathit{c}_{1}{}={}&-9 \Psi_{2} \kappa_{1}^2
 - 9 \bar\Psi_{2} \bar{\kappa}_{1'}^2
 -  \kappa_{1} \bar{\kappa}_{1'} \bigl((\overline{K_{\rho}{}} - 3 \bar{\rho}) (K_{\rho '}{} - 3 \rho ')
 + (K_{\rho}{} - 3 \rho) (\overline{K_{\rho '}{}} - 3 \bar{\rho}')\nonumber\\
& + (K_{\tau}{} - 3 \tau) (\overline{K_{\tau}{}} - 3 \bar{\tau})
 + (K_{\tau '}{} - 3 \tau ') (\overline{K_{\tau '}{}} - 3 \bar{\tau}')\bigr),\\
\mathit{c}_{a^{2}{}}{}={}&- \tfrac{81}{4} (\kappa_{1} -  \bar{\kappa}_{1'}) (\kappa_{1} + \bar{\kappa}_{1'}) (\Psi_{2} \kappa_{1}^2 -  \bar\Psi_{2} \bar{\kappa}_{1'}^2)\nonumber\\
& + \tfrac{9}{4} (\kappa_{1} -  \bar{\kappa}_{1'})^2 \bigl(\kappa_{1}^2 (K_{\rho}{} - 3 \rho) (K_{\rho '}{} - 3 \rho ') + \bar{\kappa}_{1'}^2 (\overline{K_{\rho}{}} - 3 \bar{\rho}) (\overline{K_{\rho '}{}} - 3 \bar{\rho}')\bigr)\nonumber\\
& -  \tfrac{9}{4} (\kappa_{1} + \bar{\kappa}_{1'})^2 \bigl(\kappa_{1}^2 (K_{\tau}{} - 3 \tau) (K_{\tau '}{} - 3 \tau ') + \bar{\kappa}_{1'}^2 (\overline{K_{\tau}{}} - 3 \bar{\tau}) (\overline{K_{\tau '}{}} - 3 \bar{\tau}')\bigr).
\end{align}
\end{subequations}
Technically, $\mathit{c}_{1}$ and $\mathit{c}_{a^{2}{}}$ would contain the same information if we would remove the $K$ variables from these expressions, but that would give more complicated expressions in Lemma~\ref{lem:GradientApproxConst}.

\subsection{Kerr-NUT class}
To find interpretations of the approximate constants, we study the Kerr-NUT class. This class admits a Killing spinor $\kappa$ with real $\xi$, so all the small variables vanish for this class. However, it is in general not asymptotically flat.
We do direct tetrad calculations in the principal tetrad
\begin{subequations}
\begin{align}
l^{a}={}&\frac{2 a (\partial_{\phi})^{a} + 2  (a^2 + r^2)(\partial_{v})^{a} +  (a^2 - 2 M r + r^2)(\partial_{r})^{a}}{\sqrt{2} (r^2 + a^2 x^2)},\\
n^{a}={}&- \tfrac{1}{\sqrt{2}}(\partial_{r})^{a},\\
m^{a}={}&\frac{i (\partial_{\phi})^{a} -  (1 + \frac{2 N x}{a} -  x^2)(\partial_{x})^{a} - i a  (x^2-1)(\partial_{v})^{a}}{\sqrt{2} (r - i a x) \sqrt{1 + \frac{2 N x}{a} -  x^2}}.
\end{align}
\end{subequations}
We find that the Killing spinor coefficient is given by
\begin{align}
\kappa_{1}={}&\tfrac{-1}{3} (r - i a x).
\end{align}
Observe that the radial coordinate coincides with \eqref{eq:rDef}. 
In these coordinates the Killing vectors are
\begin{align}
\xi^{a}={}&(\partial_{v})^{a},&
\zeta^{a}={}&a^2 (\partial_{v})^{a}
 + a (\partial_{\phi})^{a}.
\end{align}
The approximate constants takes the constant values
\begin{align}
\mathit{c}_{M}{}={}&M,&
\mathit{c}_{N}{}={}&-i N,&
\mathit{c}_{1}{}={}&1,&
\mathit{c}_{a^{2}{}}{}={}&a^2.
\end{align}
Hence, in general we can interpret the approximate constants as follows.
\begin{itemize}
\item $\mathit{c}_{M}$ approximates the mass. 
\item $\mathit{c}_{N}$ encodes the NUT charge, and should therefore be small for small perturbations of Kerr. 
\item $\mathit{c}_{1}$ can be interpreted as a local approximation of the square of the norm of $\xi$ at infinity. One could use this to set the correct scaling of $\kappa$. 
\item $\mathit{c}_{a^{2}}$ approximates the square of the angular momentum parameter.
\end{itemize} 
Observe that if $\kappa$ is given these expressions are local. However, to solve the elliptic equation \eqref{eq:InitialElliptic} on the initial surface $\InitSlice$, global information on $\InitSlice$ is needed.

\subsection*{Acknowledgements}
The author is grateful to Lars Andersson, Pieter Blue and Sergiu Klainerman for discussions.

\subsection*{Data Availability}
Data sharing is not applicable to this article as no new data were created or analyzed in this work.

\appendix

\section{GHP form of \texorpdfstring{$\Upsilon$}{Upsilon}, \texorpdfstring{$\gamma$}{gamma} and \texorpdfstring{$\eta$}{eta}}\label{sec:Appendix}
The GHP form of \eqref{eq:UpsilonToH} is
\begin{subequations}
\begin{align}
\frac{\Upsilon_{00'}}{\kappa_{1}}={}&(\edtp {} - 2 \tau ')\Psi_{1}
 + 3 K_{\rho}{} \Psi_{2}
 - 2 \Psi_{3} \kappa
 + \Psi_{0} \sigma ',\\
\frac{\Upsilon_{01'}}{\kappa_{1}}={}&(\thop {} - 2 \rho ')\Psi_{1}
 + 3 K_{\tau}{} \Psi_{2}
 + \Psi_{0} \kappa '
 - 2 \Psi_{3} \sigma ,\\
\frac{\Upsilon_{10'}}{\kappa_{1}}={}&(\tho {} - 2 \rho)\Psi_{3}
 + 3 K_{\tau '}{} \Psi_{2}
 + \Psi_{4} \kappa
 - 2 \Psi_{1} \sigma ',\\
\frac{\Upsilon_{11'}}{\kappa_{1}}={}&(\edt {} - 2 \tau)\Psi_{3}
 + 3 K_{\rho '}{} \Psi_{2}
 - 2 \Psi_{1} \kappa '
 + \Psi_{4} \sigma .
\end{align}
\end{subequations}
The GHP form of \eqref{eq:gammaDef} is
\begin{subequations}
\label{eq:gammaGHP}
\begin{align}
\frac{\gamma_{00'}}{\kappa_{1}}={}&- \tfrac{2}{3} (\tho {} - 4 \rho)K_{\rho}{}
 + 2 (\edtp {} - 2 \tau ')\kappa
 -  \tfrac{2}{3} K_{\rho}{}^2
 + \tfrac{2}{3} K_{\tau '}{} \kappa
 -  \tfrac{2}{3} K_{\tau}{} \bar{\kappa}
 + 2 \sigma \bar{\sigma},\\
\frac{\gamma_{01'}}{\kappa_{1}}={}&- \tfrac{1}{3} (\tho {} - 4 \rho + \bar{\rho})K_{\tau}{}
 + (\thop {} - 2 \rho ' + \bar{\rho}')\kappa
 -  \tfrac{1}{3} (\edt {} - 4 \tau + \bar{\tau}')K_{\rho}{}
 + (\edtp {} + \bar{\tau} - 2 \tau ')\sigma
 -  \tfrac{2}{3} K_{\rho}{} K_{\tau}{}\nonumber\\
& + \tfrac{1}{3} K_{\rho '}{} \kappa
 + \tfrac{1}{3} K_{\tau '}{} \sigma ,\\
\frac{\gamma_{02'}}{\kappa_{1}}={}&2 (\thop {} - 2 \rho ')\sigma
 -  \tfrac{2}{3} (\edt {} - 4 \tau)K_{\tau}{}
 -  \tfrac{2}{3} K_{\tau}{}^2
 + 2 \kappa \bar{\kappa}'
 + \tfrac{2}{3} K_{\rho '}{} \sigma
 -  \tfrac{2}{3} K_{\rho}{} \bar{\sigma}',\\
\frac{\gamma_{10'}}{\kappa_{1}}={}&- \tfrac{2}{3} (\tho {} - 3 \rho)K_{\tau '}{}
 + \tfrac{2}{3} (\edtp {} - 3 \tau ')K_{\rho}{}
 -  \tfrac{2}{3} K_{\rho '}{} \bar{\kappa}
 + \tfrac{2}{3} K_{\tau}{} \bar{\sigma},\\
\frac{\gamma_{11'}}{\kappa_{1}}={}&- \tfrac{1}{3} (\tho {} - 3 \rho + \bar{\rho})K_{\rho '}{}
 + \tfrac{1}{3} (\thop {} - 3 \rho ' + \bar{\rho}')K_{\rho}{}
 -  \tfrac{1}{3} (\edt {} - 3 \tau + \bar{\tau}')K_{\tau '}{}
 + \tfrac{1}{3} (\edtp {} + \bar{\tau} - 3 \tau ')K_{\tau}{},\\
\frac{\gamma_{12'}}{\kappa_{1}}={}&\tfrac{2}{3} (\thop {} - 3 \rho ')K_{\tau}{}
 -  \tfrac{2}{3} (\edt {} - 3 \tau)K_{\rho '}{}
 + \tfrac{2}{3} K_{\rho}{} \bar{\kappa}'
 -  \tfrac{2}{3} K_{\tau '}{} \bar{\sigma}',\\
\frac{\gamma_{20'}}{\kappa_{1}}={}&-2 (\tho {} - 2 \rho)\sigma '
 + \tfrac{2}{3} (\edtp {} - 4 \tau ')K_{\tau '}{}
 + \tfrac{2}{3} K_{\tau '}{}^2
 - 2 \bar{\kappa} \kappa '
 + \tfrac{2}{3} K_{\rho '}{} \bar{\sigma}
 -  \tfrac{2}{3} K_{\rho}{} \sigma ',\\
\frac{\gamma_{21'}}{\kappa_{1}}={}&- (\tho {} - 2 \rho + \bar{\rho})\kappa '
 + \tfrac{1}{3} (\thop {} - 4 \rho ' + \bar{\rho}')K_{\tau '}{}
 -  (\edt {} - 2 \tau + \bar{\tau}')\sigma '
 + \tfrac{1}{3} (\edtp {} + \bar{\tau} - 4 \tau ')K_{\rho '}{}\nonumber\\
& + \tfrac{2}{3} K_{\rho '}{} K_{\tau '}{}
 -  \tfrac{1}{3} K_{\rho}{} \kappa '
 -  \tfrac{1}{3} K_{\tau}{} \sigma ',\\
\frac{\gamma_{22'}}{\kappa_{1}}={}&\tfrac{2}{3} (\thop {} - 4 \rho ')K_{\rho '}{}
 - 2 (\edt {} - 2 \tau)\kappa '
 + \tfrac{2}{3} K_{\rho '}{}^2
 -  \tfrac{2}{3} K_{\tau}{} \kappa '
 + \tfrac{2}{3} K_{\tau '}{} \bar{\kappa}'
 - 2 \sigma ' \bar{\sigma}'.
\end{align}
\end{subequations}
The GHP form of \eqref{eq:etaDef} is
\begin{subequations}
\label{eq:etaGHP}
\begin{align}
\frac{\eta_{00'}}{\kappa_{1}}={}&- (\edtp {} + 2 \tau ')\kappa
 -  \tho K_{\rho}{}
 -  K_{\rho}{}^2
 - 3 K_{\tau '}{} \kappa
 -  K_{\tau}{} \bar{\kappa}
 -  \sigma \bar{\sigma},\\
\frac{\eta_{01'}}{\kappa_{1}}={}&- \tfrac{1}{2} (\tho {} + \bar{\rho})K_{\tau}{}
 -  \tfrac{1}{2} (\thop {} + 2 \rho ' + \bar{\rho}')\kappa
 -  \tfrac{1}{2} (\edt {} + \bar{\tau}')K_{\rho}{}
 -  \tfrac{1}{2} (\edtp {} + \bar{\tau} + 2 \tau ')\sigma
 -  K_{\rho}{} K_{\tau}{}
 -  \tfrac{3}{2} K_{\rho '}{} \kappa\nonumber\\
& -  \tfrac{3}{2} K_{\tau '}{} \sigma ,\\
\frac{\eta_{02'}}{\kappa_{1}}={}&- (\thop {} + 2 \rho ')\sigma
 -  \edt K_{\tau}{}
 -  K_{\tau}{}^2
 -  \kappa \bar{\kappa}'
 - 3 K_{\rho '}{} \sigma
 -  K_{\rho}{} \bar{\sigma}',\\
\frac{\eta_{10'}}{\kappa_{1}}={}&- \tfrac{2}{3} (\tho {} + \rho)K_{\tau '}{}
 -  \tfrac{2}{3} (\edtp {} + \tau ')K_{\rho}{}
 -  \tfrac{4}{3} K_{\rho}{} K_{\tau '}{}
 -  \tfrac{2}{3} K_{\rho '}{} \bar{\kappa}
 -  \tfrac{2}{3} K_{\tau}{} \bar{\sigma}
 - 4 \kappa \sigma ',\\
\frac{\eta_{11'}}{\kappa_{1}}={}&- \tfrac{1}{3} (\tho {} + \rho + \bar{\rho})K_{\rho '}{}
 -  \tfrac{1}{3} (\thop {} + \rho ' + \bar{\rho}')K_{\rho}{}
 -  \tfrac{1}{3} (\edt {} + \tau + \bar{\tau}')K_{\tau '}{}
 -  \tfrac{1}{3} (\edtp {} + \bar{\tau} + \tau ')K_{\tau}{}\nonumber\\
& -  \tfrac{2}{3} K_{\rho}{} K_{\rho '}{}
 -  \tfrac{2}{3} K_{\tau}{} K_{\tau '}{}
 - 2 \kappa \kappa '
 - 2 \sigma \sigma ',\\
\frac{\eta_{12'}}{\kappa_{1}}={}&- \tfrac{2}{3} (\thop {} + \rho ')K_{\tau}{}
 -  \tfrac{2}{3} (\edt {} + \tau)K_{\rho '}{}
 -  \tfrac{4}{3} K_{\rho '}{} K_{\tau}{}
 -  \tfrac{2}{3} K_{\rho}{} \bar{\kappa}'
 - 4 \kappa ' \sigma
 -  \tfrac{2}{3} K_{\tau '}{} \bar{\sigma}',\\
\frac{\eta_{20'}}{\kappa_{1}}={}&- (\tho {} + 2 \rho)\sigma '
 -  \edtp K_{\tau '}{}
 -  K_{\tau '}{}^2
 -  \bar{\kappa} \kappa '
 -  K_{\rho '}{} \bar{\sigma}
 - 3 K_{\rho}{} \sigma ',\\
\frac{\eta_{21'}}{\kappa_{1}}={}&- \tfrac{1}{2} (\tho {} + 2 \rho + \bar{\rho})\kappa '
 -  \tfrac{1}{2} (\thop {} + \bar{\rho}')K_{\tau '}{}
 -  \tfrac{1}{2} (\edt {} + 2 \tau + \bar{\tau}')\sigma '
 -  \tfrac{1}{2} (\edtp {} + \bar{\tau})K_{\rho '}{}
 -  K_{\rho '}{} K_{\tau '}{}\nonumber\\
& -  \tfrac{3}{2} K_{\rho}{} \kappa '
 -  \tfrac{3}{2} K_{\tau}{} \sigma ',\\
\frac{\eta_{22'}}{\kappa_{1}}={}&- (\edt {} + 2 \tau)\kappa '
 -  \thop K_{\rho '}{}
 -  K_{\rho '}{}^2
 - 3 K_{\tau}{} \kappa '
 -  K_{\tau '}{} \bar{\kappa}'
 -  \sigma ' \bar{\sigma}'.
\end{align}
\end{subequations}



%

\end{document}